\documentclass[10pt]{article}

\usepackage{amssymb}
\usepackage{amsthm}
\usepackage{amsmath}

\newtheorem{theorem}{Theorem}
\newtheorem{lemma}[theorem]{Lemma}
\newtheorem{proposition}[theorem]{Proposition}
\newtheorem{corollary}[theorem]{Corollary}

\theoremstyle{definition}
\newtheorem{definition}[theorem]{Definition}

\theoremstyle{remark}
\newtheorem{remark}[theorem]{Remark}

\numberwithin{equation}{section} \numberwithin{theorem}{section}
\def \T{\Theta}
\def \O{\Omega}
\def \tilde{\widetilde}
\def \la{\lambda}
\def \p{\partial}
\def \pa{\partial}
\def \cp{\mathbb C[\partial]}
\def \l0{L_{\geq 0}}
\renewcommand{\O}{\Omega}

\newcommand{\CC}{{\mathbb C}}

\newcommand{\ZZ}{{\mathbb Z}}

\newcommand{\E}{\mathcal{E}}
\newcommand{\M}{\mathcal{M}}
\newcommand{\F}{\mathcal{F}}
\newcommand{\A}{\mathcal{A}}

\newcommand{\R}{\mathcal{R}}

\newcommand{\fg}{{\mathfrak g}}

\newcommand{\for}{{\rm for}}

\newcommand{\Res}{{\rm Res}}
\newcommand{\Alg}{{\rm Alg}}
\newcommand{\Conf}{{\rm Conf}}

\newtheorem*{corollary*}{Corollary}
\newtheorem*{remark*}{Remark}
\newtheorem*{remarks*}{Remarks}

\newcommand{\alphaparenlist}{
\renewcommand{\theenumi}{\alph{enumi}}%
\renewcommand{\labelenumi}{(\theenumi)}%
}

\makeatletter
\def\@maketitle{\newpage
\begin{center}%
\vskip 3em
{\Large\bf \@title \par}%
\vskip 1.5em {\normalsize \lineskip .5em
\begin{tabular}[t]{c}\@author
\end{tabular}\par}%
\vskip 1em

\end{center}%
\par
\vskip .5em} \makeatother

\begin{document}

\title{Representations of simple finite Lie conformal superalgebras of type
$W$ and $S$}

\author{Carina Boyallian
\and Victor G. Kac
\and Jose I.~Liberati \and  Alexei Rudakov
}

\maketitle

\begin{abstract}  We construct all finite irreducible modules over Lie conformal
superalgebras of type $W$ and $S$.

\end{abstract}

\vskip 3cm

\section{Introduction}

Lie conformal superalgebras encode the singular part of the
operator product expansion of chiral fields in two-dimensional
quantum field theory \cite{K1}. On the other hand, they are
closely connected to the notion of formal distribution Lie
superalgebra $(\fg , \F)$, that is a Lie superalgebra $\fg$
spanned by the coefficients of a family $\F$ of mutually local
formal distributions. Namely, to a Lie conformal superalgebra $R$
one can associate a formal distribution Lie superalgebra $($Lie
$R, R)$ which establishes an equivalence between the category of
Lie conformal superalgebras and the category of equivalence
classes of formal distribution Lie superalgebras obtained as
quotients of Lie $R$ by irregular ideals \cite{K1}.

Finite simple Lie conformal algebras were classified in \cite{DK}
and all their finite irreducible representations were constructed
in \cite{CK}. According to \cite{DK}, any finite simple Lie
conformal algebra is isomorphic either to the current Lie
conformal algebra Cur$\fg$, where $\fg$ is a simple
finite-dimensional Lie algebra, or to the Virasoro conformal
algebra.

However, the list of finite simple Lie conformal superalgebras is
much richer, mainly due to existence of several series of super
extensions of the Virasoro conformal algebra. The complete
classification of finite simple Lie conformal superalgebras was
obtained in \cite{FK}. The list consists of current Lie conformal
superalgebras Cur$\fg$, where $\fg$ is a simple finite-dimensional
Lie superalgebra, four series of ``Virasoro like'' Lie conformal
superalgebras $W_n$ ($n \geq 0$), $S_{n,b}$ and $\tilde S_n$ ($n
\geq 2, b \in \CC$), $K_n$ ($n \geq 0$), and the exceptional Lie
conformal superalgebra $CK_6$.

All finite irreducible representations of the simple Lie conformal
superalgebras Cur$\fg$, $K_0 = Vir$ and $K_1$ were constructed in
\cite{CK}, and those of $S_{2,0}$, $W_1=K_2$, $K_3$, $K_4$  in
\cite {CL}.

The main result of the present paper is the construction of all
finite irreducible modules over the Lie conformal superalgebras
$W_n$, $S_{n,b}$ and $\tilde S_n$. The proof relies on the method
developed  in \cite{CK}, that is, the observation that
representation theory of  Lie conformal superalgebras is
controlled by the representation theory of the (extended)
annihilation superalgebra, and a lemma from \cite{BDK}. In our
cases, this reduces to the study of certain representations of the
Lie superalgebra $W(1,n)_+$ of all vector fields on a superline
(an affine superspace of dimension $(1|n)$) and the Lie
superalgebra $S(1,n)_+$ of such vector fields with zero
divergence. As in \cite{KR1}, \cite{KR2}, we follow the approach
developed for representations of infinite-dimensional simple
linearly compact Lie algebras by A. Rudakov in \cite{R}. The
problem reduces to the description of the so called degenerate
modules, and for the later we have to study singular vectors.

The paper is organized as follows. In Section \ref{sec:formal}, we
recall the notions and some basic facts on of formal
distributions, Lie conformal superalgebras and their modules. In
Section \ref{sec:2}, we recall some simple facts of the
representation theory of infinite-dimensional simple linearly
compact Lie superalgebras. In Section \ref{sec:3}, we describe the
conformal Lie superalgebra $W_n$ and we classify its finite
irreducible conformal modules by studying the corresponding
singular vectors. In Section  \ref{sec:4}, we obtain similar
results for the Lie conformal superalgebra $S_n = S_{n,0}$.
Finally, in Section \ref{sec:5}, we complete the cases $S_{n,b}$
and $\tilde S_n$. In all cases (as in \cite{R}) the answer has a
geometric meaning: all finite irreducible modules are either
``non-degenerate'' tensor modules, or occur as cokernels of the
differential in the conformal de Rham complex, or are duals of the
latter.

Note that similar results for arbitrary non-super Lie
pseudoalgebras of types $W$ and $S$ have been obtained in
\cite{BDK1}.

The remaining cases of the Lie conformal superalgebras $K_n$ and
$CK_6$ will be worked out in the subsequent publication.

\section{Formal distributions, Lie conformal superalgebras and their modules}\label{sec:formal}


First we introduce the basic definitions and notations, see
\cite{K1, DK}. Let $\fg$ be a Lie superalgebra. A $\fg$-valued
{\it formal distribution} in one indeterminate $z$ is a series in
the indeterminate $z$
\begin{displaymath}
a(z)=\sum_{n\in \ZZ} a_n z^{-n-1}, \qquad a_n\in \fg.
\end{displaymath}
The vector superspace of all formal distributions, $\fg[[z,
z^{-1}]]$, has a natural structure of a $\CC[\p_z]$-module. We
define
\begin{displaymath}
\Res_z a(z) =a_0.
\end{displaymath}

Let $a(z), b(z)$ be two $\fg$-valued formal distributions. They
are called $local$ if
\begin{displaymath}
(z-w)^N [a(z), b(w)]=0  \qquad \for \quad N>>0.
\end{displaymath}

Let $\fg$ be a Lie superalgebra, a  family $\F$ of $\fg$-valued
formal distributions is called a {\it local family} if  all pairs
of formal distributions from $\F$ are local. Then, the pair $(\fg
, \F)$ is called a {\it formal distribution Lie superalgebra} if
$\F$ is a local family of $\fg$-valued  formal distributions and
$\fg$ is spanned by the coefficients of all formal distributions
in $\F$. We define the {\it formal $\delta$-function} by
\begin{displaymath}
\delta(z-w)=z^{-1} \sum_{n\in\ZZ} \left(\frac{w}{z}\right)^n.
\end{displaymath}
Then it is easy to show (\cite{K1}, Corollary 2.2)), that  two
local formal distributions are local if and only if  the bracket
can be represented as a finite sum of the form
\begin{displaymath}
[a(z), b(w)]=\sum_j [a(z)_{(j)} b(w)] \  \p_w^j \delta(z-w)/j!,
\end{displaymath}
where $[a(z)_{(j)} b(w)]=\Res_z (z-w)^j[a(z), b(w)]$. This is
called the {\it operator product expansion}. Then we obtain a
family of operations $_{(n)}$, $n\in \ZZ_+$, on the space of
formal distributions. By taking the generating series of these
operations, we define the $\la$-bracket:
\begin{displaymath}
[a_\la b]=\sum_{n\in\ZZ_+} \frac{\la^n}{n!} [a_{(n)} b].
\end{displaymath}
The properties of the $\la$-bracket motivate the following
definition:

\begin{definition} A {\it   Lie conformal superalgebra} $R$ is  a left
$\ZZ/2\ZZ$-graded $\cp$-module endowed with a $\CC$-linear map,
\begin{displaymath}
R\otimes R  \longrightarrow \CC[\la]\otimes R, \qquad a\otimes b
\mapsto a_\la b
\end{displaymath}
called the $\la$-bracket, and  satisfying the following axioms
$(a,\, b,\, c\in R)$,

\

\noindent Conformal sesquilinearity $ \qquad  [\pa a_\la b]=-\la
[a_\la b],\qquad [a_\la \pa b]=(\la+\pa) [a_\la b]$,

\vskip .3cm

\noindent Skew-symmetry $\ \qquad\qquad\qquad [a_\la
b]=-(-1)^{p(a)p(b)}[b_{-\la-\pa} \ a]$,

\vskip .3cm

\noindent Jacobi identity $\quad\qquad\qquad\qquad [a_\la [b_\mu
c]]=[[a_\la b]_{\la+\mu} c] + (-1)^{p(a)p(b)}[b_\mu [a_\la c]]. $

\vskip .5cm

Here and further $p(a)\in \ZZ/2\ZZ$ is the parity of $a$.
\end{definition}

A Lie conformal superalgebra is called $finite$ if it has finite
rank as a $\CC[\pa]$-module. The notions of homomorphism, ideal
and subalgebras of a Lie conformal superalgebra are defined in the
usual way. A Lie conformal superalgebra $R$ is $simple$ if $[R_\la
R]\neq 0$ and contains no ideals except for zero and itself.

Given a  formal distribution Lie superalgebra $(\fg , \F)$ denote
by $\bar\F$ the minimal subspace of $\fg[[z, z^{-1}]]$ which
contains $\F$ and is closed under all $j$-th products and
invariant under $\p_z$. Due to Dong's lemma, we know that $\bar\F$
is a local family as well. Then $\Conf(\fg, \F):=\bar\F$ is the
Lie conformal superalgebra associated to the formal distribution
Lie superalgebra $(\fg , \F)$.

In order to give  the (more or less) reverse functorial
construction, we need the notion of {\it affinization} $\tilde R$
of a conformal algebra $R$ (which is a generalization of that for
vertex algebras \cite{B}). We let $\tilde R= R[t, t^{-1}]$ with
$\tilde\p=\p +\p_t$ and the $\la$-bracket \cite{K1}:
\begin {equation} \label{eq:0}
[a f(t)_\la bg(t)]=[a_{\la+\p_t}b] f(t)g(t')|_{t'=t}.
\end{equation}
The 0-th product is:
\begin{equation}
\label{eq:1} [a t^n_{(0)} b t^m]=\sum_{j\in\ZZ_+}
\left(\begin{array}{c} m\\j\end{array}\right) [a_{j}b] t^{m+n-j}.
\end{equation}
Observe that $\tilde\p\tilde R$ is an ideal of $\tilde R$ with
respect to the 0-th product. We let $\Alg R=\tilde R /
\tilde\p\tilde R$ with the 0-th product and let
\begin{displaymath}
\R=\{\sum_{n\in \ZZ} (at^n) z^{-n-1}= a \delta(t-z)\ |\ a\in R \}.
\end{displaymath}
Then $(\Alg R,\R)$ is a formal distribution Lie superalgebra. Note
that Alg is a functor from the category of Lie conformal
superalgebras to the category of  formal distribution Lie
superalgebras. On has \cite {K1}:
\begin{displaymath}
{\Conf}({\Alg}R)=R, \ \ {\rm Alg}({\rm Conf}(\fg,\F)) = ({\rm
Alg}\bar\F,\bar\F).
\end{displaymath}
Note also that $(\Alg R,\R)$ is the {\it maximal formal
distribution superalgebra} associated to the conformal
superalgebra $R$, in the sense that all formal distribution Lie
superalgebras $(\fg, \F)$ with $\Conf(\fg , \F)=R$ are quotients
of $(\Alg R,\R)$ by irregular ideals (that is, an ideal $I$ in
$\fg$ with no non-zero $b(z)\in\R$ such that $b_n\in I$). Such
formal distribution Lie superalgebras are called $equivalent$.

We thus have an equivalence of categories of conformal Lie
superalgebras and equivalence classes of formal distribution Lie
superalgebras. So the study of formal distribution Lie
superalgebras reduces to the study of conformal Lie superalgebras.

\

An important tool for the study of Lie conformal superalgebras and
their modules is the (extended) annihilation algebra. The {\it
annihilation algebra} of a Lie conformal superalgebra $R$  is the
subalgebra $\A( R)$ (also denoted by $(\Alg R)_+$) of the Lie
superalgebra $\Alg R$ spanned by all elements $at^n$, where $a\in
R ,  n\in \ZZ_+$. It is clear from (\ref{eq:1}) that this is a
subalgebra, which is invariant with respect to the derivation
$\p=-\p_t$ of $\Alg R$. The {\it extended annihilation algebra} is
defined as
\begin{displaymath}
\A(R)^e=(\Alg R)^+:=\CC\p \ltimes (\Alg R)_+.
\end{displaymath}
Introducing the generating series
\begin{equation} \label{eq:la}
a_\la =\sum_{j\in\ZZ_+} \frac{\la^j}{j!} (a t^j),
\end{equation}
we obtain from (\ref{eq:1}):
\begin{equation} \label{eq:2}
[a_\la , b_\mu ] = [a_\la b]_{\la +\mu}, \quad \p(a_\la)= (\p
a_\la)=\la(a_\la).
\end{equation}

Now let $\fg$ be a Lie superalgebra, and let $V$ be a
$\fg$-module. Given a $\fg$-valued formal distribution $a(z)$ and
a $V$-valued formal distribution $v(z)$ we may consider the formal
distribution $a(z)v(w)$ and the pair $(a(z), v(z))$ is called
$local$ if $(z-w)^N (a(z)v(w))=0$ for $N>>0$. As before, we have
an expansion of the form:
\begin{displaymath}
a(z)v(w)=\sum_j \Big(a(z)_{(j)} v(w)\Big) \  \p_w^j
\delta(z-w)/j!,
\end{displaymath}
where $a(w)_{(j)} v(w)=\Res_z (z-w)^ja(z)v(w)$ and the sum is
finite.  By taking the generating series of these operations, we
define the  {\it $\la $-action of $\fg$ on $V$}:
\begin{displaymath}
a(w)_\la v(w)=\sum_{n\in\ZZ_+} \frac{\la^n}{n!} \Big(a(w)_{(n)}
v(w)\Big), \qquad \hbox{ (finite sum)} .
\end{displaymath}
It has the following properties:
\begin{displaymath}
\p_z a(z)_\la v(z) =-\la  a(z)_\la v(z), \qquad a(z)_\la \p_z v(z)
= (\p_z +\la ) (a(z)_\la v(z)),
\end{displaymath}
and
\begin{displaymath}
[a(z)_\la , b(z)_\mu] v(z)=[a(z)_\la b(z)]_{\la+\mu} v(z).
\end{displaymath}
This motivate the following definition:

\begin{definition} A {\it module} M over a Lie conformal superalgebra  $R$
is  a $\ZZ/2\ZZ$-graded $\cp$-module endowed with a $\CC$-linear
map  $R\otimes M \longrightarrow \CC[\la]\otimes M$, $a\otimes v
\mapsto a_\la v$, satisfying the following axioms $(a,\, b \in
R),\ v\in M$,

\vskip .2cm

\noindent$(M1)_\la \qquad   (\pa a)_\la^M v= [\pa^M,
a_\la^M]v=-\la a_\la^Mv   ,$

\vskip .2cm

\noindent $(M2)_\la\qquad [a_\la^M, b_\mu^M]v=[a_{ \la}
b]_{\la+\mu}^Mv$.

\vskip .2cm

An $R$-module $M$ is called {\it finite} if it is finitely
generated over $\cp$. An $R$-module $M$ is called {\it
irreducible} if it contains no non-trivial submodule, where the
notion of submodule is the usual one.
\end{definition}

\

As before, if $\F\subset \fg [[z, z^{-1}]]$ is a local family and
$\E\subset V [[z, z^{-1}]]$ is such that all pairs $(a(z), v(z))$,
where $a(z)\in \F$ and $v(z)\in \E$, are local, let $\bar\E$ be
the minimal subspace of $V[[z, z^{-1}]]$ which contains $E$ and
all $a(z)_{(j)}v(z)$ for $a(z)\in \F$ and $v(z)\in \E$, and is
$\p_z$-invariant. Then it is easy to show that all pairs $(a(z),
v(z))$, where $a(z)\in \bar\F$ and $v(z)\in \bar\E$, are local and
$a(z)_{(j)}(\bar\E)\subset \bar\E$ for all $a(z)\in \bar \F$.

Let $\F$ be a local family that spans $\fg$ and let $\E\subset V
[[z, z^{-1}]]$ be a family  that span $V$. Then $(V, \E)$ is
called a {\it formal distribution module} over the formal
distribution Lie superalgebra $(\fg, \F)$ if all pairs $(a(z),
v(z))$, where $a(z)\in \F$ and $v(z)\in \E$, are local. It follows
that a formal distribution module $(V, \E)$ over a formal
distribution Lie superalgebra $(\fg, \F)$ give rise to a module
$\Conf(V,\E):=\bar \E$ over the conformal Lie superalgebra
$\Conf(\fg, \F)$.

In the same way as above, we have an equivalence of categories of
modules over a Lie conformal superalgebra $R$ and equivalence
classes or formal distribution modules over the Lie superalgebra
$\Alg R$. Namely, given an $R$-module $M$, one defines its {\it
affinization} $\tilde M=M[t, t^{-1}]$ as a $\tilde R$-module with
$\tilde\p=\p +\p_t$ and the $\la$-action similar to (\ref{eq:0}):
\begin {equation} \label{eq:3}
a f(t)_\la vg(t)=(a_{\la+\p_t}v) f(t)g(t')|_{t'=t}.
\end{equation}
The 0-th action is:
\begin{equation}
\label{eq:4} a t^n_{(0)} v t^m=\sum_{j\in\ZZ_+}
\left(\begin{array}{c} m\\j\end{array}\right) (a_{j}v) t^{m+n-j}.
\end{equation}

Observe that $\tilde\p\tilde M$ is invariant  with respect to the
0-th action and $(\tilde\p\tilde R)_{(0)}\tilde M=0$, hence the
0-th action of $\tilde R$ on $\tilde M$ induces a representation
of the Lie superalgebra $\Alg R=\tilde R / \tilde\p\tilde R$  on
$V(M):=\tilde M / \tilde\p\tilde M$. Let $\M=\{ v\delta(z-t) |
v\in M\}$. Then $(V(M), \M)$ is a formal distribution module over
the formal distribution Lie superalgebra $(\Alg R, \R)$, which is
maximal in the sense that all conformal $(\Alg R, \R)$ modules
$(V,\E)$ such that $\bar\E\simeq M$ as $R$-modules are quotients
of $(V(M), \M)$ by irregular submodules. Such formal distribution
modules are called equivalent, and we get an equivalence of
categories of $R$-modules and equivalence classes of formal
distribution $(\Alg R,\R)$-modules.

Formula (\ref{eq:2}) implies the following important proposition
relating modules over a Lie conformal superalgebra $R$ to
continuous modules over the corresponding extended annihilation
Lie superalgebra $(\Alg R)^+$.

\begin{proposition}
\label{prop:1} \cite{CK} A module over a Lie conformal
superalgebra $R$ is the same as a continuous module over the Lie
superalgebra $(\Alg R)^+$, i.e. it is an $(\Alg R)^+$-module
satisfying the property
\begin{equation} \label{eq:5}
a_\la m\in \CC[\la]\otimes M  \hbox{ \  for any } a\in R, m\in M.
\end{equation}
(One just views the action of the generating series $a_\la$ of
$(\Alg R)^+$ as the $\la$-action of $a\in R$).
\end{proposition}

Denote by $V(M)_+$ the span of elements $\{vt^n|v\in M, n\in
\ZZ_+\}$ in $V(M)$. It is clear from (\ref{eq:3}) that $V(M)^+$ is
an $(\Alg R)^+$ submodule, hence an $R$ -module by Proposition
~{\ref{prop:1}}. We denote by $V(M)^*_+$ the restricted dual of
$V(M)_+$, i.e. the space of all linear functions on $V(M)_+$ which
vanish on all but finite number of subspaces $Mt^n$, with
$n\in\ZZ_+$. This is an $(\Alg R)^+$-module and hence an
$R$-module as well. The {\it conformal dual} $M^*$ to an
$R$-module $M$ is defined as
\begin{displaymath}
M^*=\{ f_\la :M\to \CC[\la]\ |\ f_\la(\p m)=\la f_\la(m)\},
\end{displaymath}
with the structure of $\CC[\p]$-module $(\p f)_\la(m)=-\la
f_\la(m)$, with the following $\la$-action of $R$:
\begin{displaymath}
(a_\la f)_\mu(m)=-(-1)^{p(a)(p(f)+1)} f_{\mu-\la}(a_\la m), \quad
a\in R, m\in M.
\end{displaymath}

Given a homomorphism of conformal $R$-modules $T:M\to N$, we
define the transpose homomorphism $T^*:N^*\to M^*$ by
\begin{displaymath}
[T^*(f)]_\lambda (m)=-f_\la(T(m))
\end{displaymath}
\begin{proposition}\label{prop:222} Let $T: M\to N$ be an
injective homomorphism of $R$-modules such that $N/\hbox{Im }T$ is
finitely generated  torsion-free $\CC[\p]$-module. Then $T^*$ is
surjective.
\end{proposition}
\begin{proof} Since $N/\hbox{Im }T$ is
finitely generated  torsion-free, then it is free and therefore a
proyective $\CC[\p]$-module. Hence, the short exact sequence $0\to
\hbox{Im }T \to N\to N/\hbox{Im }T \to 0$ is split and $N
=\hbox{Im }T\oplus L$ as   $\CC[\p]$-module. Now, given $\alpha\in
M^*$, we define $\beta\in N^*$ as follows
$$
\beta_\la (T(m)) =\alpha_\la (m), \quad m\in M,\qquad
\beta_\la(l)=0, \quad \ l\in L.
$$
Then $\beta$ is well-defined since $T$ is injective and $\beta$
belong to $N^*$ since $L$ is a complementary $\CC[\p]$-submodule,
finishing the proof.
\end{proof}
\begin{remark}\label{rm:22} Observe that the injectivity is not
enough (cf. Remark~\ref{rm:12}). Namely, let $R=Vir=\CC[\p]L$ the
Virasoro conformal algebra with $\la$-bracket $[L_\la L]= (2\la
+\p)L$. Consider the following $Vir$-modules:
$$
\O_0=\CC[\p]m, \ \hbox{ with } L_\la m = (\la+\p)m;\qquad
\O_1=\CC[\p]n, \ \hbox{ with } L_\la n = \p n.
$$
Then it is easy to see that the map $d: \O_0 \to \O_1$ given by
$d(m)=\p n$ is an injective homomorphism of $R$-modules, but the
dual map $d^*: \O_1^* \to \O_0^*$ given by $d^*(m^*)=\p n^*$ is
not surjective.
\end{remark}
\begin{proposition}\label{prop:22} Let $T: M\to N$ be a
 homomorphism of $R$-modules such that $N/\hbox{Im }T$ is
finitely generated  torsion-free $\CC[\p]$-module. Then the
standard map $\psi:N^*/\hbox{Ker }T^* \to (M/\hbox{Ker }T)^*$,
given by $[\psi(\bar f)]_\la(\bar m)=f_\la(T(m))$ (where by the
bar we denote the corresponding class in the quotient) is an
isomorphism of $R$-modules.
\end{proposition}
\begin{proof} Using Proposition~\ref{prop:222} the proof follows by standard arguments.
\end{proof}
\begin{proposition}
\label{prop:dd} If $M$ is an $R$-module finitely generated (over
$\cp$), then $M^{**}\simeq M$.
\end{proposition}
\begin{proof} Let $M=\oplus \ \cp m_i$ (finite sum), with $a_\la
m_j=\sum_k P_{jk}(\la, \p)m_k$. Then $M^*=\oplus \  \cp m_i^*$,
with $(m_i^*)_\la (m_k)=\delta_{i,k}$ and
$$
(a_\la m_i^*)_\mu(m_j)= - (m_i^*)_{\mu-\la}(a_\la m_j)= -\sum_k
(m_i^*)_{\mu-\la} (P_{jk}(\la, \p)m_k)=P_{ji}(\la, \mu-\la).
$$
Therefore,
$$
(a_\la m_i^*)=  -\sum_j  P_{ji}(\la, -\p-\la)  m_j^*,
$$
and the last formula shows that by taking the dual again we obtain
$$
(a_\la m_i^{**})=  \sum_j  P_{ij}(\la, \p)  m_j^{**}.
$$
Hence the map $m_i\mapsto m_i^{**}$ gives us the isomorphism
between $M$ and $M^{**}$.
\end{proof}

\begin{proposition}
\label{prop:2} (a) The map $M\to V(M)/V(M)_+$ given by $v\mapsto
vt^{-1}$ mod $V(M)_+$ is an isomorphism of $(\Alg R)^+$- (and
$R$-)modules.

\noindent (b)The map $V(M)^*_+\to M^*$ defined by $f\mapsto
f_\la$, where
\begin{displaymath}
f_\la(m)=\sum_{j\in\ZZ_+} \frac{(-\la)^j}{j!} f(mt^j),
\end{displaymath}
is an isomorphism of $(\Alg R)^+$- (and $R$-)modules.
\end{proposition}

\begin{proof}
A direct verification.
\end{proof}

\

Assuming that $R$ is finite, choose a finite set of generators of
this $\CC[\p]$-module: $\{a^i | i\in I\}$, and for each $m\in
\ZZ_+$, denote by $(\Alg R)^+_{(m)}$ the $\CC$-span of all
elements $a^i t^j, i\in I, j\geq m$ of $(\Alg R)^+$. This defines
a descending filtration of $(\Alg R)^+$ by subspaces of finite
codimension:
\begin{equation} \label{eq:6}
(\Alg R)^+\supset(\Alg R)_+=(\Alg R)_{(0)}\supset(\Alg
R)_{(1)}\supset\dots.
\end{equation}
It is easy to see from (\ref{eq:1}) that there exists $s\in \ZZ_+$
such that for all $k,r\in\ZZ_+$ one has:
\begin{equation} \label{eq:7}
[(\Alg R)_{(k)} , (\Alg R)_{(r)}]\subset (\Alg R)_{(k+r-s)} .
\end{equation}
In particular, $(\Alg R)_{r}:=(\Alg R)_{(r+s)}$ is a filtration of
$(\Alg R)^+$ by subalgebras of finite codimension.

Given an $R$-module $M$, it is an $(\Alg R)^+$-module (by
Proposition~{\ref{prop:1}}), and we let for $j\in \ZZ_+$:
\begin{equation} \label{eq:8}
M_{(j)}=\{v\in M | (\Alg R)_{(j)} v =0\}.
\end{equation}
The subspaces $M_{(j)}$ form an ascending filtration of $M$ by
$(\Alg R)_0$-invariant subspaces. The following proposition is a
special case of Lemma 14.4 from \cite{BDK}.

\begin{proposition}
\label{prop:3} let $R$ be a finite Lie conformal superalgebra and
let $M$ be a finite $R$-module such that
\begin{displaymath}
M^R=\{ m\in M | R_\la m=0\} (=M_{(0)})
\end{displaymath}
is finite dimensional (over $\CC$). Then all subspaces $M_{(j)}$
are finite-dimensional. In particular $M$ is locally finite as an
$(\Alg R)_0$-module, meaning that any $m\in M$ is contained in a
finite-dimensional $(\Alg R)_0$-invariant subspace.
\end{proposition}

This Proposition together with the results of the following
section will provide a characterization of all finite irreducible
modules over a finite Lie conformal superalgebra in terms of
certain  (quotients of) induced modules over the extended
annihilation algebra.

%
%

\section{General remarks on representations of linearly compact Lie
superalgebras}\label{sec:2}

We follow  the approach developed for representations of
infinite-dimensional simple linearly compact Lie algebras by A.
Rudakov in \cite{R}. In this section we will follow \cite{KR1}.

We shall consider continuous representations in spaces with
discrete topology. The continuity of a representation of a
linearly compact Lie superalgebra $L$  in a vector space $V$ with
discrete topology means that the stabilizer $L_v=\{ g\in L |
gv=0\}$ of any $v\in V$ is an open (hence of finite codimension)
subalgebra of $L$.

Let $L$ be a simple linearly compact Lie superalgebra. In some
cases (the examples studied in the following sections), $L$ has a
$\ZZ$-gradation of the form
\begin{equation} \label{eq:9}
L=\oplus_{m\geq -1} L_m,
\end{equation}
this gives a triangular decomposition
\begin{equation} \label{eq:10}
L=L_- \oplus L_0 \oplus L_+,\qquad \hbox{  with } \quad
L_\pm=\oplus_{\pm m>0} L_m.
\end{equation}
Let $L_{\geq 0}=L_{ > 0} \oplus L_0$. Denote by $P(L,L_{\geq 0})$
the category  of all continuous $L$-modules $V$, where $V$ is a
vector space with discrete topology, that are $\l0$-locally
finite, that is any $v\in V$ is contained in a finite-dimensional
$\l0$-invariant subspace. When talking about representations of
$L$, we shall always mean modules from $P(L,\l0)$. Modules in this
category are called {\it finite continuous $L$-modules}.

In general, in most (but not all cases) of simple $L$, by taking
$\l0$ certain maximal open subalgebra, one can choose $L_-$ to be
a subalgebra. Taking an ordered basis of $L_-$, we denote by
$U(L_-)$ the span of all PBW monomials in this basis. We have
$U(L)=U(L_-)\otimes U(\l0)$, as vector spaces (here and further
$U(L)$ stands for the universal enveloping algebra of the Lie
superalgebra $L$). It follows that any irreducible $L$-module $V$
in the category $P(L,\l0)$ is finitely generated over $U(L_-)$:
\begin{displaymath}
V=U(L_-)E
\end{displaymath}
for some finite-dimensional subspace $E$. This property is very
important in the theory of conformal modules \cite{CK}.

Given an $\l0$-module $F$, we may consider the associated induced
$L$-module
\begin{displaymath}
M(F)=\hbox{ Ind}^L_{\l0} F=U(L)\otimes_{U(\l0)} F,
\end{displaymath}
called the {\it generalized Verma module} associated to $F$.
Sometimes, we shall omit $L$ and $\l0$, and simply denote it as
Ind$\,F$.

Let $V$ be an $L$-module. The elements of the subspace
\begin{displaymath}
\hbox{ Sing}(V):=\{ v\in V| L_{>0} v=0\}
\end{displaymath}
are called {\it singular vectors}. For us the most important case
is when $V=M(F)$. The $\l0$-module $F$ is canonically an
$\l0$-submodule of $M(F)$, and Sing$(F)$ is a subspace of
Sing$(M(F))$, called the {\it subspace of trivial singular
vectors}. Observe that $M(F)= F\oplus F_+$, where
$F_+=U_+(L_-)\otimes F$ and $U_+(L_-)$ is the augmentation ideal
in the symmetric algebra $U(L_-)$. Then
\begin{displaymath}
\hbox{ Sing}_+(M(F)):=\hbox{ Sing}(M(F))\cap F_+
\end{displaymath}
are called the {\it non-trivial singular vectors}.

\begin{theorem}
\label{th:1} \cite{KR1}\cite{R} (a) If $F$ is a finite-dimensional
$\l0$-module, then  $M(F)$ is in $P(L,\l0)$.

(b) In any irreducible finite-dimensional $\l0$-module $F$ the
subalgebra $L_+$ acts trivially.

(c) If $F$ is an irreducible finite-dimensional $\l0$-module, then
$M(F)$ has a unique maximal submodule.

(d) Denote by I(F) the quotient  by the unique maximal submodule
of $M(F)$. Then the map $F\mapsto I(F)$ defines a bijective
correspondence between irreducible finite-dimensional
$L_0$-modules and irreducible $L$-modules in $P(L,\l0)$, the
inverse map being $V\mapsto $Sing$(V)$.

(e) An $L$-module $M(F)$ is irreducible if and only if the
$L_0$-module $F$ is irreducible and $M(F)$ has no non-trivial
singular vectors.
\end{theorem}

\begin{remark} \label{rmk:1}
The correspondence defined in Theorem~{\ref{th:1}}(d) provides the
classification of irreducible modules of the category $P(L,\l0)$.
Also, we would like to remark that in general Sing$_+(M(F))$
generates a proper submodule in the $L$-module $M(F) $, but the
factor by this submodule is not necessarily irreducible, there
could appear new non-trivial singular vectors. However this
happens very rarely (see \cite{KR2} for an example and cf. Remark
\ref{rm:8}) and in most cases it can be proven that the factor
module will be irreducible.
\end{remark}

%
%

\section{Lie conformal superalgebra $W_n$ and its finite irreducible
modules}\label{sec:3}

%
%

\subsection{Definition of $W_n$ and the induced modules}\label{subsec:w1}

According to \cite{DK}, any finite simple Lie conformal algebra is
isomorphic either to Cur$\fg$, where $\fg$ is a simple
finite-dimensional Lie algebra, or to the Virasoro conformal
algebra.

However, the list of finite simple Lie conformal superalgebras is
much richer, mainly due to existence of several series of super
extensions of the Virasoro conformal algebra, see \cite{FK}.

The first series is associated to the Lie superalgebra $W(1,n)$
$(n\geq 1)$. More precisely, let $\Lambda(n)$ be the Grassmann
superalgebra in the $n$ odd indeterminates $\xi_1, \xi_2,\ldots ,
\xi_n$. Set $\Lambda(1,n)=\CC[t, t^{-1}]\otimes \Lambda(n)$, then
\begin{equation} \label{eq:3.1}
W(1,n)=\{a\p_t +\sum_{i=1}^n a_i \p_i | a, a_i\in\Lambda(1,n)\},
\end{equation}
where $\p_i=\frac{\p}{\p\xi_i}$ and $\p_t=\frac{\p}{\p t}$ are odd
and even derivations respectively. Then $W(1,n)$ is a formal
distribution Lie superalgebra with spanning family of (pairwise
local) formal distributions:
\begin{displaymath}
\F=\{\delta(t-z) a \ |\ a\in W(n)\}\cup\{\delta(t-z)f\p_t\ | \
f\in\Lambda (n)\}.
\end{displaymath}
where $W(n)=\{ \sum_{i=1}^n a_i \p_i |   a_i\in\Lambda(n)\}$ is
the (finite-dimensional) Lie superalgebra of all derivations of
$\Lambda(n)$. The associated Lie conformal superalgebra $W_n$ is
defined as
\begin{equation} \label{eq:3.2}
W_n=\CC[\p]\otimes \left(W(n)\oplus\Lambda(n)\right).
\end{equation}
The $\la$-bracket is defined as follows $(a,b\in W(n); f,g\in
\Lambda(n))$:
\begin{equation} \label{eq:3.3}
[a_\la b]=[a,b], \quad [a_\la f]= a(f)-(-1)^{p(a)p(f)}\la fa,\quad
[f_\la g]=-(\p +2\la )fg
\end{equation}
The Lie conformal algebra $W_n$ is simple for $n\geq 0$ and has
rank $(n+1)2^n$.

The annihilation subalgebra is

\begin{equation} \label{eq:3.4}
\A(W_n)=W(1,n)_+ =\{a\p_t +\sum_{i=1}^n a_i \p_i | a,
a_i\in\Lambda(1,n)_+\},
\end{equation}
where $\Lambda(1,n)_+=\CC[t]\otimes \Lambda(n)$. The extended
annihilation subalgebra is
\begin{displaymath}
\A(W_n)^e= W(1,n)^+=\CC \p_t \ltimes W(1,n)_+,
\end{displaymath}
and therefore it is isomorphic to the direct sum of $W(1,n)_+$ and
a commutative 1-dimensional Lie algebra.

The $\ZZ$-gradation  in (\ref{eq:9}) is obtained by letting
\begin{displaymath}
\hbox{ deg }t=\hbox{ deg }\xi_i =1=-\hbox{ deg }\p_t=-\hbox{ deg
}\p_i.
\end{displaymath}
If $L=W(1,n)_+$, then $L_{-1}=<\p_t,\p_1,\ldots ,\p_n>$, where
$\p_t$ is an even element  and $\p_1,\ldots ,\p_n$ are odd
elements of a basis in $L_{-1}$. Note also that $L_0\simeq
gl(1|n)$.

From now on, we shall use the notation $\p_0=\p_t$. Explicitly, we
have
\begin{displaymath}
L_0=< \{ t\p_i, \xi_i \p_j \ : \ 0\leq i,j\leq n\} > .
\end{displaymath}
In order to write explicitly weights for vectors in
$W(1,n)_+$-modules, we would consider the basis
\begin{displaymath}
t\p_0 ; t\p_0 +\xi_1 \p_1 , \ldots , t\p_0 +\xi_n \p_n
\end{displaymath}
for the Cartan subalgebra $H$ in $W(1,n)_+$, and we write the
weight of an eigenvector for the Cartan subalgebra $H$ as a tuple
\begin{displaymath}
\bar\mu=(\mu; \la_1 , \ldots , \la_n )
\end{displaymath}
for the corresponding eigenvalues of the basis.

%
%

\subsection{Modules of Laurent differential forms}\label{subsec:w2}

\

\noindent\ref{subsec:w2}.1 {\it Restricted dual}. Our algebra
$L=W(1,n)_+$, and in the last section $S(1,n)_+$, are $\ZZ$-graded
(super)algebras and the modules we intend to study are graded
modules, i.e. an $L$-module $V$ is a direct sum $V=\oplus_{m\in
\ZZ} V_m$ of finite-dimensional subspaces $V_m$, and $L_k \cdot
V_m \subset V_{k+m}$. For a graded module $V$ we define the {\it
restricted dual module} $V^\#$ as
\begin{displaymath}
V^\# =\oplus_{m\in \ZZ} (V_m)^*.
\end{displaymath}
hence $V^\#$ is a subspace of $V^*$ and it is invariant with
respect to the contragradient action, so it defines an $L$-module
structure. Observe that $(V^\#)^\#=V$.

In our situation, we have $L_{-1}=\langle \p_0,\p_1, \ldots
,\p_n\rangle$, then any $L$-module become a $\CC[\p_0,\p_1, \ldots
,\p_n]$-module. Hence, a module $V$ is a free $\CC[\p_0,\p_1,
\ldots ,\p_n]$-module if and only if $V^\#$ is a {cofree} module,
i.e. it is isomorphic to a direct sum of copies of the standard
module $\CC[z, \rho_1, \ldots , \rho_n]$, with $\p_0 \cdot f=
\frac{\p}{\p z} f$, and $\p_i \cdot f= \frac{\p}{\p \rho_i} f$.

An induced module Ind$^L_{\l0} F$ is by definition a free
$\CC[\p_0,\p_1, \ldots ,\p_n]$-module, so the co-induced (or
produced) module
\begin{displaymath}
\hbox{Cnd} F^\#= (\hbox{Ind}F)^\#
\end{displaymath}
will be cofree.

\

\noindent\ref{subsec:w2}.2 {\it Differential forms modules}. In
order to define the differential forms one considers an odd
variable $dt$ and even variables $d\xi_1, \ldots , d\xi_n$ and
defines the differential forms to be the (super)commutative
algebra freely generated by these variables over
$\Lambda(1,n)_+=\CC[t]\otimes \Lambda(n) $, or
\begin{displaymath}
\Omega_+=\Lambda(1,n)_+ [d\xi_1, \ldots , d\xi_n]\otimes
\Lambda[dt].
\end{displaymath}
Generally speaking $\Omega_+$ is just a polynomial (super)algebra
over a big set of variables
\begin{displaymath}
t, \xi_1, \ldots , \xi_n, dt, d\xi_1, \ldots , d\xi_n,
\end{displaymath}
where the parity is
\begin{displaymath}
p(t)=0, \ \ p(\xi_i)=1, \ \ p(dt)=1, \ \ p(d\xi_i)=0.
\end{displaymath}
These are called {\it (polynomial) differential forms}, and we
define the {\it Laurent differential forms} to be the same algebra
over $\Lambda(1,n)=\CC[t, t^{-1}]\otimes \Lambda(n)$:
\begin{displaymath}
\Omega=\Lambda(1,n) [d\xi_1, \ldots , d\xi_n]\otimes \Lambda[dt].
\end{displaymath}
We would like to consider a fixed complementary subspace
$\Omega_-$ to $\Omega_+$ in $\Omega$ chosen as follows
\begin{displaymath}
\Omega_- = t^{-1}\CC[t^{-1}]\otimes \Lambda(n)\otimes\CC [d\xi_1,
\ldots , d\xi_n]\otimes \Lambda[dt].
\end{displaymath}

For the differential forms we need the usual differential degree
that measure only the involvement of the differential variables
$dt, d\xi_1, \ldots , d\xi_n$, that is
\begin{displaymath}
\hbox{ deg } t=0, \hbox{ deg } \xi_i=0, \hbox{ deg } dt=1, \hbox{
deg } d\xi_i=1.
\end{displaymath}
As a result, the degree of a function is zero an it gives us the
{\it standard $\ZZ$-gradation} both on $\O$ and $\O_\pm$. As
usual, we denote by $\O^k, \O^k_\pm$ the corresponding graded
components.

We denote by $\O^k_c$ the special subspace of differential forms
with constant coefficients in $\O_k$.

The operator $d$ is defined on $\O$ as usual by the rules $ d
\cdot t = dt,  d \cdot \xi_i = d\xi_i, d \cdot d\xi_i = 0$, and
the identity
\begin{displaymath}
d(fg)=(df) g + (-1)^{p(f)} f dg,
\end{displaymath}
Observe that $d$ maps both $\O_+$ and $\O_-$ into themselves.

As usual, we extend the natural action  of $W(1,n)_+$ on
$\Lambda(1,n)$ to the hole $\O$ by imposing the property
\begin{displaymath}
D \cdot d=(-1)^{p(D)} d\cdot D,\qquad D\in W(1,n)_+,
\end{displaymath}
that is, $D$ (super)commutes with $d$. It is clear that $\O_+$ and
all the subspaces $\O^k$ are invariant. Hence $\O_+^k$ and  $\O^k$
are $W(1,n)_+$-modules, which are called the {\it natural
representations} of $W(1,n)_+$ in differential forms.

We define the action of $W(1,n)_+$ on $\O_-$ via the isomorphism
of $\O_-$ with the factor of $\O$ by $\O_+$. Practically this
means that in order to compute $D\cdot f$, where $f\in \O_-$, we
apply $D$ to $f$ and "disregard terms with non-negative powers of
$t$".

The operator $d$ restricted to $\O^k_\pm$ defines an odd morphism
between the corresponding representations. Clearly the image and
the kernel of such a morphism are submodules in $\O^k_\pm$.

\

Let $\T^k_c= (\O^k_c)^\#$ and  $\T^k_+= (\O^k_+)^\#$. In the rest
of this subsection, we consider $ L=W(1,n)_+$.

\

\begin{proposition} \label{prop:d1} (1) The $L_0$-module $\T^k_c, k\geq 0$
is irreducible with highest weight
\begin{displaymath}
(0; 0,\ldots ,0,-k),\ \  k\geq 0.
\end{displaymath}
\vskip .3cm

(2) The $L$-module $\T^k_+,\  k\geq 0$ contains $\T^k_c$ and this
inclusion induces the isomorphism
\begin{displaymath}
\T^k_+=\hbox{\rm Ind}\ \T^k_c.
\end{displaymath}
\vskip .3cm

(3) The dual maps $d^\# : \T^{k+1}_+ \to \T^k_+$ are morphisms of
$L$-modules. The kernel of one of them is equal to the image of
the next one and it is a non-trivial proper submodule in $\T^k_+$.

\vskip .3cm

\end{proposition}

\begin{proof}
(1) It is well known that $\O^k_c$ are irreducible and thus
$\T^k_+$ are also irreducible. Observe that  the lowest vector in
$\O^k_c$ is $(d\xi_n)^k$ and it has the weight $(0;0,\ldots ,
0,k)$. Now the sign changes as we go to the dual module and so we
get the highest weight of $\T^k_c$.

(2) By the definition of the restricted dual, it is the sum of the
dual of all the graded components of the initial module. In our
case $\O^k_c$ is the component of the minimal degree in $\O^k_+$,
so $\T^k_c$ becomes the component of the maximal degree in
$\T^k_+$. This implies that $L_{>0}$ acts trivially on $\T^k_c$,
so the morphism Ind$\ \T^k_c \to \T^k_+$ is  defined. Clearly
$\O^k_+$ is isomorphic to
\begin{displaymath}
\O^k_c\otimes \CC [t,\xi_1, \ldots , \xi_n],
\end{displaymath}
so it is a cofree module. Then the module $\T^k_+$ is a free
$\CC[\p_0, \p_1, \ldots ,\p_n]$-module and the morphism
\begin{displaymath}
\hbox{\rm Ind}\ \T^k_c \to \T^k_+
\end{displaymath}
is therefore an isomorphism.

(3) This statement follows immediately from the fact that $d$
commutes with the action of vector fields.
\end{proof}

\begin{corollary}\label{cor:d} The $L$-modules $\O^k_+$ of differential
forms are isomorphic to the co-induced modules
\begin{displaymath}
\O^k_+=\hbox{\rm Cnd}\ \O^k_c.
\end{displaymath}
\end{corollary}

Let us now study the $L$-modules $\O^k_-$. First, notice that
these modules are free as $\CC[\p_0, \p_1,\ldots , \p_n]$-modules.
Let
\begin{equation} \label{eq:omega}
\xi_* = \xi_1\cdots \xi_n, \qquad \hbox{ and } \qquad \bar\O^k_c =
t^{-1}\xi_* \O^k_c \subset \O^k_-.
\end{equation}

\begin{proposition} \label{prop:d2} For $L=W(1,n)_+$, we have:

\vskip .3cm

(1) $\bar\O^k_c$ is an irreducible $L_0$-submodule in $\O^k_-$
with highest  weight
\begin{eqnarray*}
(-1;0,0,\ldots , 0),  \  \hbox{ for } k=0,
\\
(0;k,1,\ldots , 1),  \  \hbox{ for } k>0,
\end{eqnarray*}
and $L_{>0}$ acts trivially on $\bar\O^k_c$.

\vskip .3cm

(2) There is an isomorphism $\O^k_- =\hbox{\rm Ind}^L_{\l0}
\bar\O^k_c$.

\vskip .3cm

(3) The differential $d$ gives us  $L$-module morphisms on
$\O^k_-$ and the kernel and image of $d$ are $L$-submodules in
$\O^k_-$.

\vskip .3cm

(4) The kernel of $d$ and image of $d$ in $\O^k_-$ for $k\geq 2$
coincide,  in $\O^1_-$ we have {\rm Ker}$\ d=\CC(t^{-1} dt) +
${\rm Im}$\ d$, and in $\O^0_-$, we have {\rm Ker}$\ d=0$ (and the
image does not exist).

\vskip .3cm
\end{proposition}

\begin{proof}
(1) First of all, $\bar\O^k_c$ is the maximum total degree
component in $\O^k_-$, so any element from $L_{>0}$ moves it to
zero. Also, as $L_0$-module $\bar\O^k_c$ is isomorphic to $\O^k_c$
multiplied by the 1-dimensional module $\langle
t^{-1}\xi_*\rangle$. This permits us to see that its highest
vectors are
\begin{eqnarray*}
\langle t^{-1}\xi_*\rangle \qquad\qquad \ \ \hbox{ for }k=0,
\\
\langle t^{-1}\xi_* dt\rangle \qquad\qquad \hbox{ for }k=1,
\\
\langle t^{-1}\xi_*dt (d\xi_1)^{k-1}\rangle \qquad  \hbox{ for
}k>1.
\end{eqnarray*}
The values of the highest weights are easy  to compute.

(2) It is straightforward to see that $\O^0_-$ is a free rank 1
$\CC[\p_0,\p_1, \ldots , \p_n]$-module. Now, the action of
$\p_0,\p_1, \ldots , \p_n$ on $\O^k_-$ is coefficient wise and the
fact that  $\O^k_-$ is a free $\CC[\p_0,\p_1, \ldots ,
\p_n]$-module follows. This gives us the isomorphism
$\O^k_-=$Ind$^L_{\l0}\bar\O^k_c$. Parts (3) and (4) are left to
the reader.
\end{proof}

The above statement shows us that there are non-trivial submodules
in $\O^k_\pm$ and $\T^k_+$. In fact, these are "almost all" proper
submodules and the respective factors are irreducible. These
results are discussed in Section~{\ref{subsec:w4}}. In order to
get this result we need to study singular vectors.

%
%

\subsection{Singular vectors  of $W_n$-modules}\label{subsec:w3}

\

Having in mind the results of Section \ref{sec:2}, we introduce
the following modules. Given a $gl(1|n)$-module $V$, we have the
associated tensor field $W(1,n)$-module $\CC[t,
t^{-1}]\otimes\Lambda(n)\otimes V$, which is a formal distribution
module spanned by a collection of fields
$E=\{\delta(t-z)fv|f\in\Lambda(n), v\in V\}$. The associated
conformal $W_n$-module is
\begin{equation} \label{eq:3.5}
\hbox{ Tens}(V)=\CC[\p]\otimes\left(\Lambda(n)\otimes V)\right)
\end{equation}
with the following $\la$-action:
\begin{eqnarray}  \label{eq:3.6}
& a_\la(g\otimes v) = a(g)\otimes v +(-1)^{p(a)}\sum_{i,j=1}^n (\p_i f_j)g\otimes
(E_{ij}-\delta_{ij})(v) - \hskip 2cm
\\
& \hskip 6cm - \la(-1)^{p(g)} \sum_{j=1}^n f_j g \otimes
E_{0j}(v), \nonumber
\end{eqnarray}
\begin{eqnarray} \label{eq:3.7}
& f_\la(g\otimes v) = (-\p )(fg\otimes v) +(-1)^{p(fg)} \sum_{i=1}^n (\p_i
f)g\otimes E_{i0}(v) + \hskip 1cm
\\
& \hskip 6cm + \la ( f g \otimes E_{00}(v)). \nonumber
\end{eqnarray}
where $a=\sum_{i=1}^n f_i \p_i \in W(n), f,g\in\Lambda(n), v\in
V$, and $E_{ij}\in gl(1|n)$ are matrix units (they correspond to
the level 0 elements $\xi_i\p_j$ with the notation $\xi_0=t$ and
$\p_0=\p_t$).

In this case, the modules $M(F)=\hbox{ Ind}_{\l0}^L F$ defined in
Section \ref{sec:2}, correspond to the $W_n$-module Tens$(F)$,
with $F$ a finite-dimensional (irreducible)  $gl(1|n)$-module.
When we discuss the highest weight of vectors and singular
vectors, we always mean with respect to the upper Borel subalgebra
in $L=W(1,n)_+$ generated by $L_{>0}$ and the elements of $L_0$:
\begin{equation} \label{eq:borel}
t\p_i, \quad \xi_i \p_j \ \ i<j.
\end{equation}

Therefore, in the module $M(V)$, viewed as a module over the
annihilation algebra $W(1,n)_+$ (see Proposition~{\ref{prop:1}}),
a vector $m\in M(V)$ is a singular vector if and only if the
following conditions are satisfied ($g=\xi_{i_1}\cdots
\xi_{i_s}\in \Lambda(n)$, and $\p_0=\p_t$)

\

(s1) \ $t^n g \p_i \cdot m=0$ for $n>1$,

\vskip .3cm

(s2) \ $t^1 g \p_i \cdot m=0$ except for $g=1$ and $i=0$,

\vskip .3cm

(s3) \ $t^0 g \p_j \cdot m=0$  for $s>1$ or $g=\xi_i$ with $i<j$.
\begin{equation} \label{eq:s3}
\end{equation}


We shall frequently use the notation
\begin{equation} \label{eq:xi}
\xi_I=\xi_{i_1}\cdots \xi_{i_s}\in \Lambda(n), \quad \hbox{ with }
I=\{i_1,\ldots , i_s\}.
\end{equation}
Therefore, these conditions on a singular vector $m\in $Tens$(V)$
translate in terms of the $\la$-action  to (cf. (\ref{eq:la})):

\

(S1) $\frac{d^2}{d\la^2}(f_\la m)=0$ for all $f\in\Lambda(n)$,

\vskip .3cm

(S2) $\frac{d}{d\la}(a_\la m)=0$ for all $a\in W(n)$,

\vskip .3cm

(S3) $\frac{d}{d\la}(f_\la m)|_{\la=0}=0$ for all $f\in\Lambda(n)$
with $f\neq 1$,

\vskip .3cm

(S4) $(a_\la m)|_{\la=0}=0$ for all $a=\xi_I\p_j\in W(n)$ with
$|I|>1$ or $a=\xi_i\p_j$ with $i<j$,

\vskip .3cm

(S5) $(f_\la m)|_{\la=0}=0$ for all $f=\xi_I\in\Lambda(n)$ with
$|I|> 1$.

\

In order to classify the finite irreducible $W_n$-modules we
should solve these equations (S1-5) to obtain the singular
vectors.

Let $m\in \hbox{Tens}(V)=\CC[\p]\otimes \Lambda(n)\otimes V$, then
\begin{equation} \label{eq:s1}
m=\sum_{k=0}^N\sum_I \p^k (\xi_I\otimes v_{I,k}), \quad \hbox{
with } v_{I,k}\in V.
\end{equation}

In order to obtain the singular vectors, we need the some
reduction lemmas:

\begin{lemma} \label{lem:deg} If $m\in \hbox{\rm Tens}(V)$ is a singular
vector, then the degree of $m$ in $\p$ is at most 1.
\end{lemma}

\begin{proof} Using (\ref{eq:3.6}) , we have for $a=\sum_{i=1}^n f_i \p_i$ that
\begin{eqnarray} \label{eq:a}
&(a_\la m)'=\sum_{k=1}^N\sum_I k (\la +\p)^{k-1}
\bigg[a(\xi_I)\otimes
v_{I,k} \qquad\qquad\qquad \qquad\qquad\qquad\\
\nonumber
& + (-1)^{p(a)} \sum_{i,j=1}^n (\p_i f_j)\xi_I\otimes
(E_{ij}-\delta_{ij})(v_{I,k})
- \la (-1)^{|I|} \sum_{j=1}^n f_j \xi_I \otimes E_{0j}(v_{I,k})\bigg]\\
\nonumber
&- \sum_{k=0}^N \sum _I (\la +\p)^k (-1)^{|I|} \sum_{j=1}^N f_j
\xi_I \otimes E_{0j}(v_{I,k}).
\end{eqnarray}
Taking $a=\p_j$, condition (S2) become
\begin{eqnarray} \label{eq:3.8}
&0=\sum_{k=1}^N\sum_{I| j\in I} k (\la +\p)^{k-1} (\xi_{i_1}\cdots
\hat\xi_j\cdots \xi_{i_s}\otimes v_{I,k} )\qquad\qquad\qquad\qquad\\
\nonumber
& - \la \sum_{k=1}^n\sum_I (-1)^{|I|} k (\la +\p)^{k-1}  ( \xi_I
\otimes
E_{0j}(v_{I,k}))\\
\nonumber
&- \sum_{k=0}^N \sum _I (\la +\p)^k (-1)^{|I|}  (\xi_I \otimes
E_{0j}(v_{I,k})).
\end{eqnarray}
Now, viewed as a polynomial in $\la$, we obtain
\begin{equation} \label{eq:3.9}
E_{0j}(v_{I,k})=0, \qquad \forall I, k=1,\ldots , N , \ \hbox{ and
} j=1, \ldots , n.
\end{equation}
Using it in (\ref{eq:3.8}) and taking the coefficients in $\la
+\p$, we get
\begin{displaymath}
v_{I,k}=0 \qquad \hbox{ for all } I\neq \emptyset, \hbox{ and }
k\geq 2.
\end{displaymath}
Hence, $m=\sum_{k=0}^1\sum_I \p^k (\xi_I\otimes v_{I,k}) +
\sum_{k=2}^N  \p^k (1\otimes v_{\emptyset ,k})$.

Using (\ref{eq:3.7}) and condition (S1) for $f=1$, we have
\begin{eqnarray} \label{eq:3.10}
& 0=(f_\la m)"=  2\sum_{I}  (\xi_{I} \otimes E_{00}(v_{I,1}) )
\qquad\qquad \qquad\\
\nonumber
&
- \sum_{k=2}^N  (k-1) k (\la +\p)^{k-2} \p ( 1 \otimes  v_{\emptyset ,k})\\
\nonumber
&+ \sum_{k=2}^N  \bigg(2k (\la +\p)^{k-1 } + \la k (k-1)(\la
+\p)^{k-2 } \bigg) (1\otimes E_{00}(v_{\emptyset ,k}))
\end{eqnarray}
Then, viewed as a polynomial in $\la$, we have
$E_{00}(v_{\emptyset ,k})$ for all $k\geq 2$. Hence the last term
in (\ref{eq:3.10}) is 0. Now, viewed as a polynomial in
$(\la+\p)$, we obtain $v_{\emptyset ,k}=0$ for all $k\geq 2$,
finishing the proof.
\end{proof}

Observe that the coefficient in $(\la+\p)^0$ in (\ref{eq:3.10}),
gives us the following useful identity
\begin{equation} \label{eq:e}
E_{00}(v_{I,1})=0 \qquad \hbox{ for all } I.
\end{equation}

We will use the following notation: $[1,n]=\{1,\ldots , n\}$.
\begin{lemma} \label{lem:deg2} If $m$ is a singular vector, then
\begin{displaymath}
m= \p (\xi_{[1,n]} \otimes w) + \sum_{l=1}^n
(\xi_{[1,n]-\{l\}}\otimes v_{l }) + \xi_{[1,n]} \otimes v_0.
\end{displaymath}
\end{lemma}

\begin{proof} By the previous Lemma, we have
\begin{displaymath}
m=  \sum_I \bigg[\p (\xi_I\otimes v_{I,1}) +  (\xi_I\otimes
v_{I,0})\bigg]
\end{displaymath}
Now (S5) gives us

\begin{eqnarray} \label{eq:3.11}
& 0=(f_\la m)_{|_{\la=0}} = \hskip 7cm
\\
\nonumber
&  =\sum_{I }  \bigg( (-\p) (f\xi_{I} \otimes v_{I,0} ) -
(-1)^{|I|} \sum_{i=1}^n (\p_i f) \xi_{I} \otimes E_{i0}(v_{I,0})
\bigg) \qquad\qquad
\qquad\\
\nonumber
&+\sum_{I }  \bigg( (-\p^2) (f\xi_{I} \otimes v_{I,1} ) -
(-1)^{|I|} \sum_{i=1}^n \p ((\p_i f) \xi_{I} \otimes
E_{i0}(v_{I,1})) \bigg)
\end{eqnarray}
for any $f=\xi_J$ with $|J|>1$. Considering the coefficient of
$\p^2$ and taking $f=\xi_l\xi_k$, we obtain $v_{I,1}=0$ for all
$I$ with $|I|\leq n-2$. Using this and considering the coefficient
of $\p$ with $f=\xi_l\xi_k\xi_s$, we obtain $v_{I,0}=0$ for all
$I$ with $|I|\leq n-3$. With this reduction, the coefficient of
$\p$ with $f=\xi_i\xi_j$ $(i \neq j)$ is
\begin{displaymath}
0= - (\xi_{[1,n]} \otimes v_{[1,n]-\{i,j\},0} ) + (-1)^{n-1}
(\xi_{[1,n]} \otimes E_{i0}(v_{[1,n]-\{j\},1}) ),
\end{displaymath}
obtaining
\begin{equation} \label{eq:3.12}
E_{i0}(v_{[1,n]-\{j\},1}) = (-1)^{n-1} v_{[1,n]-\{i,j\},0}\qquad
\hbox{ for all }i\neq j .
\end{equation}

Computing (S3), we have
\begin{eqnarray} \label{eq:3.13}
\nonumber
& 0=(f_\la m)'_{|_{\la=0}} = \hskip 8cm
\\
\nonumber
& =\sum_{|I|\geq n-1 } \bigg( (-\p) (f\xi_{I} \otimes v_{I,1} ) -
(-1)^{|I|} \sum_{i=1 }^n  (\p_i f) \xi_{I} \otimes E_{i0}(v_{I,1}) \bigg) \\
%
%
& + \ \p\ \sum_{|I|\geq n-1}    f \xi_{I} \otimes E_{00}(v_{I,1})
+
 \sum_{|I|\geq n-2}    f
\xi_{I} \otimes E_{00}(v_{I,0}) .\nonumber
\end{eqnarray}
Now, taking $f=\xi_i$, using (\ref{eq:e}) and considering the
coefficient in $\p$, we have
\begin{displaymath}
v_{I,1}=0 \qquad \hbox{ for all }|I|=n-1  ,
\end{displaymath}
and using it in (\ref{eq:3.12}), we have
\begin{displaymath}
v_{I,0}=0 \qquad \hbox{ for all }|I|=n-2  ,
\end{displaymath}
finishing the proof.
\end{proof}

Let $\xi_*:=\xi_{[1,n]}$ and $\xi^l:=\xi_{[1,n]-\{l\}}$. Due to
the previous lemma, any singular vector has the form
\begin{displaymath}
m= \p (\xi_* \otimes w) + \sum_{l=1}^n   (\xi^l\otimes v_{l }) +
\xi_* \otimes v_0.
\end{displaymath}
Then,  it is easy to see that conditions (s1-3) are equivalent  to
the following list

\

\noindent (s1):
\begin{eqnarray} \label{eq:v1}
E_{00}(w)=0, \qquad \qquad \qquad \quad \
\\
\label{eq:v2}
E_{0i}(w)  =0, \qquad    i=1,\ldots ,n.
\end{eqnarray}
\noindent(s2):
\begin{eqnarray} \label{eq:v3}
E_{ji}(w)+E_{0i}(v_j)=0, & \qquad     i,j=1,\ldots ,n ,
\\
\label{eq:v4}
E_{0i}(v_0)=0, & \qquad     i=1,\ldots ,n ,
\\
\label{eq:v5}
E_{0i}(v_j)=0, & \qquad          i,j=1,\ldots ,n; \ \ , i\neq j ,
\\
\label{eq:v6}
E_{0j}(v_j)= -w, & \qquad      j=1,\ldots ,n.
\\
\label{eq:v7}
E_{i0}(w)=E_{00}(v_i), & \qquad     i=1,\ldots ,n ,
\end{eqnarray}
\noindent (s3):
\begin{eqnarray}\label{eq:v8}
E_{i0}(v_j)= E_{j0}(v_i), & \qquad     i,j=1,\ldots ,n; \ \ i \neq
j.
\\
\label{eq:v9}
E_{ij}(v_l)= E_{lj}(v_i), & \qquad     i,j,l=1,\ldots ,n; \ \ i
\neq l.
\\
\label{eq:v10}
E_{ij}(w)=0, & \qquad    i,j=1,\ldots ,n; \ \ i< j ,
\\
\label{eq:v11}
E_{ij}(v_0)=0, & \qquad    i,j=1,\ldots ,n; \ \ i < j ,
\\
\label{eq:v12}
E_{ij}(v_l)=0, &\qquad  i,j,l=1,\ldots ,n; \ \ i < j, l\neq j ,
\\
\label{eq:v13}
E_{ij}(v_j)= v_i, & \qquad     i,j=1,\ldots ,n; \ \ i < j.
\end{eqnarray}

Now replacing (\ref{eq:v5}) and (\ref{eq:v6}) on (\ref{eq:v3}), we
obtain
\begin{equation} \label{eq:v14}
E_{ij}(w) = \delta_{ij} \ w,  \qquad     i,j=1,\ldots ,n .
\end{equation}

Recall that we are considering the basis $(\p_0=\p_t)$
\begin{displaymath}
t\p_0 ; t\p_0 +\xi_1 \p_1 , \ldots , t\p_0 +\xi_n \p_n
\end{displaymath}
for the Cartan subalgebra $H$ in $W(1,n)_+$, and we write the
weight of an eigenvector for the Cartan subalgebra $H$ as a tuple
\begin{equation} \label{eq:mu}
\bar\mu=(\mu; \la_1 , \ldots , \la_n )
\end{equation}
for the corresponding eigenvalues of the basis.

Using the above conditions, we can prove the following

\begin{proposition} \label{prop:s1}
Let $n\geq 2$ and $m$ be a non-trivial singular vector in {\rm
Tens}$\ V$ with weight $\bar\mu_m$, then we have one of the
following:

\vskip .3cm

(a) $m=\xi^n \otimes v_n$, $\bar\mu_m=(0;0,\ldots ,0, -k)$ with
$k\geq 0$,   $v_n$ is a highest weight vector in $V$ with weight
$(0;0,\ldots ,0, -k-1)$, and  $m$ is uniquely defined by $v_n$.

\vskip .3cm

(b) $m=\sum_{l=1}^n \xi^l \otimes v_l$, $\bar\mu_m=(0;k,1, \ldots
,1)$ with $k\geq 2$,   $v_1$ is a highest weight vector in $V$
with weight $(0;k-1,1,\ldots ,1)$, and $m$ is uniquely defined by
$v_1$.

\vskip .3cm

(c) $m=\p (\xi_* \otimes w) + \sum_{l=1}^n \xi^l \otimes v_l$,
$\bar\mu_m=(-1;0, \ldots ,0)$,   $w$ is a highest weight vector in
$V$ with weight $(0;1,\ldots ,1)$, and $m$ is uniquely defined by
$w$.

\vskip .5cm

%
%
%
%

\end{proposition}
\begin{proof}
By computing $E_{00}\cdot m = (t\p)\cdot m$ and  using
(\ref{eq:v1}) and (\ref{eq:v7}) on it, we obtain the following
conditions:

\noindent If $w=0$, then
\begin{equation} \label{eq:3.41}
E_{00}\cdot m=0 \quad \hbox{ and }\quad  E_{00}(v_0)=0.
\end{equation}
If $w\neq 0$, then
\begin{eqnarray}  \label{eq:3.42}
E_{00}\cdot m =-m &
\\
\label{eq:43} E_{00}(v_l) =-v_l, & \qquad   l=0, \ldots , n,
\end{eqnarray}
and using (\ref{eq:v7}), in this case $(w\neq 0)$ we have
\begin{equation} \label{eq:3.43}
E_{i0}(w)=-v_i,  \qquad   i=1, \ldots , n.
\end{equation}

Similarly, observe that $E_{ii}\cdot m = (\xi_i\p_i)\cdot m $. Now
this action can be easily computed, and using (\ref{eq:v14}) on
it, we have the following:

If $w\neq 0$, then
\begin{eqnarray}  \label{eq:3.45}
E_{ii}\cdot m =m, \qquad \quad \qquad  i=1, \ldots , n,
\\
\nonumber E_{ii}(v_l)=  v_l,  \qquad     l,i=1, \ldots , n, l\neq
i,
\\
\nonumber E_{ii}(v_i)=  2v_i,  \qquad  \quad \qquad  i=1, \ldots ,
n.
\\
\nonumber E_{ii}(v_0)=  v_0,  \qquad \quad  \qquad   i=1, \ldots ,
n.
\end{eqnarray}
Using this and equations (\ref{eq:3.42}), (\ref{eq:v1}) and
(\ref{eq:v14}), we obtain for the case $w\neq 0$, that the
corresponding weights are
\begin{displaymath}
\bar\mu_m =(-1; 0,  \ldots , 0) \qquad \hbox{ and }\qquad
\bar\mu_w =(0; 1,  \ldots , 1).
\end{displaymath}
This result together with (\ref{eq:3.43}) give us the proof of
case (c) in the proposition.

\

For the rest of the proof, we assume $w=0$, let us show that the
only possible cases are (a) and (b).

Observe that replacing  (\ref{eq:v13}) in (\ref{eq:v9}), we get
\begin{equation} \label{eq:3.46}
E_{jj}(v_i)=v_i \qquad \forall i<j.
\end{equation}

Now,  equation (\ref{eq:v13}) shows that if $v_l=0$ for some $l$
with $1\leq l\leq n$, then $v_j=0$ for all $j<l$. In order to
finish the proof, we should show that only the two extreme cases
are possible, that is $v_l\neq 0$ for all $l$ or $v_l=0$ except
for $l=n$.

Now, suppose that there exist $l>1$ such that $v_j=0$ for all
$j<l$ and $v_l\neq 0$, then using (\ref{eq:v9}) we have that
\begin{equation} \label{eq:3.aa}
E_{ii}(v_k)=0,  \qquad   i<l\leq k.
\end{equation}
Then by (\ref{eq:3.aa}) and (\ref{eq:3.46}), we obtain that
\begin{displaymath}
E_{ii} \cdot m = \alpha \ m \quad \hbox{with } \alpha=0  \ \hbox{
or }\  1, \qquad \ \hbox{ if }\  i<l \ \hbox{ or }\ i>l, \ \hbox{
respectively}.
\end{displaymath}
Therefore, using this and (\ref{eq:3.41}) we get
\begin{displaymath}
\bar\mu_m =(0;0, \ldots, 0, k, 1, \ldots, 1)
\end{displaymath}
where $E_{ll} \cdot m = k \ m$. But the space $V$, from which we
are inducing  is finite-dimensional and a singular vector
generates a finite-dimensional $L_0$-submodule, then (recall
notation (\ref{eq:mu}))
\begin{displaymath}
\la_1\geq \la_2\geq \cdots \geq \la_n
\end{displaymath}
is a highest weight, and because of that  only two extreme
positions of $k$ are possible (recall that $n>1$). This gives us
the cases (a) and (b). In order to finish the proof we need to
complete the computation of  weights in each case.

If $v_1\neq 0$, then using (\ref{eq:v7}) and (\ref{eq:3.46}) we
obtain
\begin{displaymath}
\bar\mu_m =(0; k,1,  \ldots , 1) \qquad \hbox{ and }\qquad
\bar\mu_{v_1} =(0; k-1,1, \ldots , 1), \qquad \hbox{ with } k\geq
1.
\end{displaymath}
getting case (b).

If $v_l= 0$ except for $l=n$, then using (\ref{eq:v7}) and
(\ref{eq:3.aa}) we obtain
\begin{displaymath}
\bar\mu_m =(0; 0,  \ldots ,0, k) \qquad \hbox{ and }\qquad
\bar\mu_{v_n} =(0; 0, \ldots , 0 , k-1), \qquad \hbox{ with }
k\leq 0.
\end{displaymath}
obtaining case (c). Case (d) is immediate.
\end{proof}
%

%
%

\subsection{Irreducible induced $W(1,n)_+$-modules}\label{subsec:w4}

In this subsection we consider $L=W(1,n)_+$, with $n\geq 2$. Now,
we have the following:

\begin{theorem} \label{th:r1} Let $n\geq 2$ and $F$ be an irreducible
$L_0$-module  with highest weight $\bar\mu_*$. Then the
$L$-modules {\rm Ind}$_{\l0}^L F$ are irreducible finite
continuous modules except for the following cases:

\vskip .3cm

(a) $\bar\mu_*=(0; 0,\ldots ,0,-m), m\geq 0$, where {\rm
Ind}$_{\l0}^L F= \Theta^m_+$ and the image $d^\#\Theta_+^{m+1}$ is
the only non-trivial proper submodule.

\vskip .3cm

(b) $\bar\mu_*=(0;k ,1,\ldots ,1), k\geq 1$, where {\rm
Ind}$_{\l0}^L F= \Omega^k_-$. For $k\geq 2$ the image $d
\Omega_-^{k-1}$ is the only non-trivial proper submodule. For
$k=1$, both {\rm Im}$(d)$ and {\rm Ker}$(d)$ are proper
submodules. {\rm Ker}$(d)$ is a maximal submodule.

\vskip .3cm

\end{theorem}
\begin{remark} \label{rm:8} Let $F$ be an irreducible
$L_0$-module  with highest weight $\bar\mu_*=(-1; 0,\ldots ,0)$.
Then  {\rm Ind}$_{\l0}^L F= \Omega^0_-$ is an irreducible
$L$-module. Note that the image of $d: \Omega^0_- \to \Omega^1_-$
is the submodule in $\Omega^1_-$ generated by the singular vector
correponding to the case (c) in Proposition~\ref{prop:s1}, but it
is not a maximal submodule (see Proposition~\ref{prop:d2} (4)).
\end{remark}
\vskip .3cm

\begin{proof}
We know from Theorem~{\ref{th:1}} that in order for
$\hbox{Ind}_{\l0}^L F$ to be reducible it has to have non-trivial
singular vectors and the possible highest weights of $F$ in this
situation are listed in Proposition~{\ref{prop:s1}} above.

The fact that the induced modules are actually reducible in those
cases is known because we have got nice realizations for these
induced modules in Propositions~{\ref{prop:d1}} and
~{\ref{prop:d2}} together with morphisms defined by $d, d^\#$, so
kernels and images of these morphisms become submodules.

The subtle thing is to prove that a submodule is really a maximal
one. We notice that in each case the factor is isomorphic to a
submodule in another induced module so it is enough to show that
the submodule is irreducible. This can be proved as follows, a
submodule in the induced module is irreducible if it is generated
by any highest singular vector that it contains. We see from our
list of non-trivial singular vectors that there is at most one
such a vector for each case and the images and kernels in question
are exactly generated by those vectors, hence they are
irreducible.
\end{proof}

\begin{corollary}\label{cor:r2} The theorem gives us a description of
finite continuous irreducible $W(1,n)_+$-modules for $n\geq 2$.
Such a module is either $\hbox{\rm Ind}_{\l0}^L F$ for an
irreducible finite-dimensional $L_0$-module $F$ where the highest
weight of $F$ does not belong to the types listed in (a), (b) of
the theorem or the factor of an induced module from (a), (b)  by
its submodule $Ker(d)$.
\end{corollary}

%
%

\subsection{Finite irreducible $W_n$-modules}\label{subsec:w5}

In order to give an explicit construction and classification, we
need the following notation. Recall that $W(1,n)$ acts by
derivations on the algebra of differential forms $\O=\O(1,n)$, and
note that this is a conformal module by taking the family of
formal distributions
\begin{displaymath}
E=\{\delta(z-t)\omega \hbox{ and }\delta(z-t)\omega \ dt \ | \
\omega\in \O(n)\}
\end{displaymath}
Translating this and all other attributes of differential forms,
like de Rham differential, etc. into the conformal algebra
language, we arrive to the following definitions.

Recall that given an algebra $A$, the associated current formal
distribution algebra is $A[t,t^{-1}]$ with the local family
$F=\{a(z)=\sum_{n\in \ZZ} (a t^n) z^{-n-1} = a \delta(z-t)\}_{a\in
A}$. The associated conformal algebra is Cur$A=\CC [\p]\otimes A$
with multiplication defined by $a_\la b=ab$ for $a,b\in A$ and
extended using sesquilinearity. This is called the {\it current
conformal algebra}, see \cite{K1} for details.

The conformal algebra of differential forms $\O_n$ is the current
algebra over the commutative associative superalgebra $\O(n)
+\O(n)\ dt$ with the obvious multiplication and parity, subject to
the relation $(dt)^2=0$:
\begin{displaymath}
\O_n=\hbox{Cur}(\O(n) +\O(n)\ dt).
\end{displaymath}
The de Rham differential $\tilde d$ of $\O_n$ (we use the tilde in
order to distinguish it from the de Rham differential $d$ on
$\O(n)$) is a derivation of the conformal algebra $\O_n$ such
that:
\begin{equation} \label{eq:d1}
\tilde d(\omega_1 +\omega_2 dt)=d\omega_1 + d\omega_2 dt -
(-1)^{p(\omega_1)} \p (\omega_1 dt).
\end{equation}
here and further $\omega_i\in \O(n)$.

The standard $\ZZ_+$-gradation $\O(n)=\oplus_{j\in \ZZ_+} \O(n)^j$
of the superalgebra of differential forms by their degree induces
a $\ZZ_+$-gradation
\begin{displaymath}
\O_n=\oplus_{j\in \ZZ_+} \O_n^j, \qquad \hbox{ where }
\O_n^j=\CC[\p]\otimes (\O(n)^j +\O(n)^{j-1}\ dt),
\end{displaymath}
so that $\tilde d : \O^j_n\to \O_n^{j+1}$.

The contraction $\iota_D$ for $D=a+f\in W_n$ is a conformal
derivation of $\O_n$ such that:
\begin{eqnarray}  \label{eq:d2}
\nonumber &(\tilde L_a)_\la (\omega_1 +\omega_2 dt) = L_a\omega_1
+ (L_a\omega_2) dt,
\\
&(\tilde L_f)_\la \omega = -(\p +\la)(f\omega),
\\
&(\tilde L_f)_\la (\omega dt) = (-1)^{p(f)+p(\omega)} (df)\omega -
\p(f \omega dt). \nonumber
\end{eqnarray}

The properties of $\O(1,n)$ imply the corresponding properties of
$\O_n$ given by the following proposition.

\begin{proposition} \label{prop:d}
\alphaparenlist
\begin{enumerate}
\item 
$\tilde d^2=0$.
\item 
The complex $(\O_n, \tilde d)=\{ 0\to \O_n^0 \to \cdots \to \O_n^j
\to \cdots\} $ is exact at all terms $\O_n^j$, except for $j=1$.
One has: {\rm Ker} $\tilde d_{|_{\O_n^1}} = \tilde d \O_n^0 \oplus
\CC dt$.
\item 
$\iota_{D_1} \iota_{D_2} + p(D_1, D_2) \iota_{D_2}\iota_{D_1} =0$.
\item 
$\tilde L_D\tilde d = (-1)^{p(D)} \tilde d \tilde L_D$.
\item 
$\tilde L_D= \tilde d \iota_D + (-1)^{p(D)} \iota_D \tilde d$.
\item 
The map $D\mapsto \tilde L_D$ defines a $W_n$-module structures on
$\O_n$, preserving the $\ZZ_+$-gradation and commuting with
$\tilde d$.
\end{enumerate}
\end{proposition}
\begin{proof} Only the proof of (b) requires a comment. Following
Proposition 3.2.2 of \cite{K1}, we construct $\CC[\p]$-linear maps
$K:\O_n \to \O_n$ (a homotopy operator) and $\epsilon :\O_n \to
\O_n$ by the formulas $(\omega\in \O(n) +\O(n) dt)$:
\begin{eqnarray*}
&K(d \xi_n \omega) =\xi_n \omega, \quad K(\omega)=0 \qquad
\hbox{ if $\omega$ does not involve } d\xi_n,
\\
&\epsilon (d \xi_n \omega) =\epsilon(\xi_n \omega)=0 , \quad
\epsilon (\omega)=\omega  \qquad \hbox{ if $\omega$ does not
involve both $ d\xi_n$ and $\xi_n$}.
\end{eqnarray*}
One checks directly that
\begin{displaymath}
K \tilde d +\tilde d K=1-\epsilon.
\end{displaymath}
Therefore, if $\omega\in \O_n$ is a closed form, we get $\omega=
\tilde d(K\omega)+ \epsilon(\omega)$. It follows  by induction on
$n$ that $\omega =\tilde d \omega_1 + P(\p) dt$ for some
$\omega_1\in\O_n$ and a polynomial $P(\p)$. But it is clear from
(\ref{eq:d1}) that $P(\p) dt$ is always closed, and it is exact
iff $P(\p)$ is divisible by $\p$.
\end{proof}

Since the extended annihilation algebra $W(1,n)^+$ is a direct sum
of $W(1,n)_+$ and a 1-dimensional Lie algebra $\CC a$, any
irreducible $W(1,n)^+$-module is obtained from a $W(1,n)_+$-module
$M$ by extending to $W(1,n)^+$, letting $a\mapsto -\alpha$, where
$\alpha\in\CC$. Translating into the conformal language (see
Proposition~{\ref{prop:1}}), we see that all $W_n$-modules are
obtained from conformal $W(1,n)_+$-modules by taking for the
action of $\p$ the action of $-\p_t +\alpha I, \alpha\in\CC$. We
denote by Tens$_\alpha V$ and $\O_{k,\alpha}, \alpha\in \CC$, the
$W_n$-modules obtained from Tens$V$ and $\O_k$ by replacing in
(\ref{eq:3.6})  and (\ref{eq:3.7}) respectively $\p$ by $\p
+\alpha$.

Now, Theorem~{\ref{th:r1}} and Corollary~{\ref{cor:r2}}, along
with Section~{\ref{sec:2}} and Propositions~{\ref{prop:1}},
~{\ref{prop:2}}, ~{\ref{prop:22}} and  ~{\ref{prop:3}} give us a
complete description of finite irreducible $W_n$-modules.

\

\begin{theorem} \label{th:w}
The following is a complete list of non-trivial finite irreducible
$W_n$-modules $(n\geq 2, \alpha\in \CC )$:
\alphaparenlist
\begin{enumerate}
\item 
{\rm Tens}$_\alpha V$, where $V$ is a finite-dimensional
irreducible $gl(1|n)$-module different from
$\Lambda^k(\CC^{1|n})^*$,  $k=1,2, \ldots $ and $\bar\Omega_c^k$
(see (\ref{eq:omega})), $k=1,2,...$,
\item 
$\O_{k,\alpha}^*/\hbox{\rm Ker }\tilde d^*, k=1,2,\ldots$ , and
the same modules with reversed parity,
\item 
$W_n$-modules dual to $(b)$, with $k>1$.
\end{enumerate}
\end{theorem}

\begin{remark} \label{rm:12} (a) Using Proposition~\ref{prop:d2}, we have that the kernel of $\tilde d$ and the image
of $\tilde d$ coincide in $\O_k$ for $k\geq 2$. Now, since
$\O_{k+2}$ is a free $\CC[\p]$-module of finite rank and
$\O_{k+1}/{\rm Im }\tilde d=\O_{k+1}/{\rm Ker }\tilde d\simeq {\rm
Im }\tilde d\subset \O_{k+2}$, we obtain that $\O_{k+1}/{\rm Im }\
\tilde d$ is a finitely generated free $\CC[\p]$-module.
Therefore, we can apply
 Proposition~\ref{prop:22}, and  we have that
\begin{equation} \label{eq:4.88}
\O_{k+1,\alpha}^*/\hbox{\rm Ker }\tilde d^*\simeq
\bigl(\O_{k,\alpha}/\hbox{\rm Ker }\tilde d\bigr)^*
\end{equation} for $k\geq 1$.

 (b) Observe that  we can not apply the previous argument for $k=0$ since, by Proposition~\ref{prop:d2}, the
 image of $\tilde d$ has codimension one (over $\CC$) in Ker $\tilde d$. In fact, (\ref{eq:4.88}) is not true for
 $k=0$. For example, this can be easily seen for $W_0=Vir$ using the differential map which is  explicitly written
  in Remark~\ref{rm:22}.

 (c) Observe that $\O_{0,\alpha}$ is an irreducible tensor module
(Ker $\tilde d=0$, cf. Proposition~\ref{prop:d2}), that is why
this module is included in case (a) of Theorem~\ref{th:w}.

(d) Since for a finite rank module $M$ over a Lie conformal
superalgebra we have $M^{**}=M$ (see Proposition~\ref{prop:dd}),
the $W_n$-modules in case (c) of Theorem~\ref{th:w} are isomorphic
to $\Omega_{k,\alpha}/ker \ \tilde d$, $k=2,3,...$.

(e) Observe that $(\rm{Tens }V)^*$ is not isomorphic to $\rm{Tens
}V^*$. For example, consider the case of $W_1$. We have, using the
notation below, that $M(a,b)=$Tens $V_{a,b}$. It is
easy to see that for the case $a+b\neq 0$, $(\rm{Tens
}V_{a,b})^*=$Tens $V_{-a,-b}$, but $(V_{a,b})^*= V_{1-a, -b-1}$.

\end{remark}

\

Now we will present the case $n=1$ in detail and we shall see that
our result agrees with the classification given in \cite{CL} for
$K_2\simeq W_1$. Let us fix some notations. We have
\begin{displaymath}
W_1=\CC[\p]\otimes \left(\Lambda(1)\oplus W(1)\right)=\CC[\p]\{1,
\xi, \p_1, \xi\p_1\}.
\end{displaymath}
In \cite{CL}, the conformal Lie superalgebra $K_2$ is presented as
the  freely generated module over $\CC[\p]$ by $\{L, J, G^\pm\}$.
An isomorphism between $K_2$ and $W_1$ is explicitly given by
\begin{equation}
L\mapsto -1+\frac 1 2 \p \xi\p_1, \qquad J\mapsto \xi\p_1,  \qquad
G^+\mapsto 2\xi, \qquad G^-\mapsto -\p_1.
\end{equation}

The irreducible modules of  $W_1$ are parameterized by
finite-dimensional irreducible representations of $gl(1,1)$ (and
the additional twist by alpha that, for simplicity, shall be
omitted  in the formulas below). The irreducible representations
of $gl(1,1)$, denoted by $V_{a,b}$, are parameterized by  a and b,
the corresponding  eigenvalues of $e_{11}$ and $e_{22}$ on the
highest weight vector.

If both parameters are equal to zero, the representation is
trivial 1-dimensional. Otherwise, either $a+b=0$, the dimension
of the  $gl(1,1)$-representation  is 1, and the corresponding
representation  of $W_1$ is one of the tensor modules of rank 2.
Or else $a+b$ is non-zero, the dimension of the
$gl(1,1)$-representation is 2, and the corresponding tensor module
has rank 4.


Explicitly, consider the set of $\Bbb C[\partial]$-generators of
$W_1$ $\{1,\,\xi,\,\partial_1,\,\xi\partial_1\}$. Let $a$ and $b$
such that $a+b\neq 0$. Let $V_{a,b}=\CC$-span$\{v_0,v_1\}$, where
$v_0 $ is a highest weight vector. Let $M(a,b)=M(V_{a,b})=\Bbb
C[\partial]\{v_0,\,v_1,\,w_1=\partial_1v_0,\, w_0=\partial_1v_1\}$
be the tensor $W_1$-module and denote by $L(a,b)$ the irreducible
quotient. The action of $W_1$ in $M(a,b)$ is given explicitly by
the following formulas:
\begin{eqnarray}\label{eq:mab}
1_\lambda v_0 = (a\lambda-\partial)v_0, & 1_\lambda v_1 =
((a-1)\lambda-\partial)v_1, \cr 1_\lambda w_1 =
(a\lambda-\partial)w_1, & 1_\lambda w_0 =
((a-1)\lambda-\partial)w_0, \nonumber
\end{eqnarray}
\begin{eqnarray}
\xi_\lambda v_0 =v_1, \quad\qquad\qquad\quad &\xi_\lambda v_1 =
0,\qquad\qquad \quad\cr \xi_\lambda w_1 =
(a\lambda-\partial)v_0-w_0, &\quad \xi_\lambda w_0 =
((a-1)\lambda-\partial)v_1, \nonumber
\end{eqnarray}
\begin{eqnarray}
{\partial_1}_\lambda v_0 = w_1, &\quad\quad {\partial_1}_\lambda
v_1 = (a+b)\lambda v_0+w_0,\cr {\partial_1}_\lambda w_1 =0,
&\quad\quad {\partial_1}_\lambda w_0 = -(a+b)\lambda w_1\nonumber
\end{eqnarray}
\begin{eqnarray}
{\xi\partial_1}_\lambda v_0 = b\,v_0, \qquad &
{\xi\partial_1}_\lambda v_1 = (b+1)\,v_1,\quad\cr
{\xi\partial_1}_\lambda w_1 = (b-1)w_1, &\quad\quad
{\xi\partial_1}_\lambda w_0 = -(a+b)\lambda\,v_0+b\,w_0.
\end{eqnarray}

If $a+b\neq 0$ and $a\neq 0$, then $M(a,b)$ is irreducible of rank
4, and the explicit action is given by (4.45). The proof of this
statement is in the following way (the proof of the other
statements below are much simpler): Take $v=p(\p) v_0 + q(\p)w_0 +
r(\p) v_1 + s(\p)w_1$ in a submodule of $M(a,b)$. Denote by $w$
the coefficient of the highest power in $\la$ of $\xi_\la v$ and
by $y$ the coefficient of the highest power in $\la$ of $\xi_\la
w$.

If $a\neq 1$ then $y=v_1$ (up to a constant factor), therefore
$v_1$ lies in  the submodule. If $a=1$, then by taking the
coefficient of the highest power in $\la$ of $\xi{\p_1}_\la y$ and
using that in this case $b\neq -1$, we also obtain that $v_1$ lies
in the submodule.

Therefore, in any case we have that $v_1$ lies in any submodule,
and by the formulas for the actions on $v_1$ it is immediate that
the other generators also belong to any submodule, proving that
$M(a,b)$ is irreducible in this case.

If $a+b\neq 0$ but $a=0$, it is easy to show as above that $N=\Bbb
C[\partial]w_1\oplus\Bbb C[\partial](\partial v_0+w_0)$ is a
submodule of $M(0,b)$. The irreducible quotient of $M(0,b)$ by $N$
is $L(0,b)=\Bbb C[\partial]v_0\oplus \Bbb C[\partial]v_1$, of rank
2,  and the action here is explicitly, as follows:
\begin{eqnarray}
1_\lambda v_0 = (-\partial)v_0, & 1_\lambda v_1 =
(-\lambda-\partial)v_1,\cr \xi_\lambda v_0 = v_1, & \xi_\lambda
v_1 = 0,\cr {\partial_1}_\lambda v_0 = 0, & {\partial_1}_\lambda
v_1 = (b\lambda-\partial)v_0,\cr {\xi\partial_1}_\lambda v_0 =
b\,v_0, & {\xi\partial_1}_\lambda v_1 = (b+1)\,v_1.
\end{eqnarray}

If  $a+b=0$, but $a\neq 0$, it is easy to show as above that
$M(a,-a)= C[\partial]\{v_0,w_1\}$ is irreducible of rank 2 and the
action of $W_1$ here is given by:

\begin{eqnarray}
1_\lambda v_0 = (a\lambda-\partial)v_0, & 1_\lambda w_1 =
(a\lambda-\partial)w_1,\cr \xi_\lambda v_0 = 0, & \xi_\lambda w_1
= (a\lambda-\partial)v_0,\cr {\partial_1}_\lambda v_0 = w_1, &
{\partial_1}_\lambda w_1 = 0,\cr {\xi\partial_1}_\lambda v_0 =
-a\,v_0, & {\xi\partial_1}_\lambda w_1 = (-a-1)w_1.
\end{eqnarray}

Thus we obtain

\begin{corollary} The $W_1$-module $L(a,b)$ as a $\Bbb C[\partial]$-module has
rank: 4 if $a+b\neq 0$ and $a\neq 0$, 2 if $a+b\neq 0$ and $a=0$,
2 if $a+b=0$ and $a\neq 0$, 0 if $a=b=0$. These are all
non-trivial finite irreducible $W_1$-modules.
\end{corollary}

\begin{remark} In \cite{CL}, the irreducible representations of $K_2$ are classified in terms of parameters $\Lambda$ and $\Delta$. Using the isomorphism between $K_2$ and $W_1$ in (4.44), these parameters are related to ours as follows,
\begin{displaymath}
a= -\Delta-\frac{\Lambda}{2}\, , \,b=\Lambda.
\end{displaymath}
Then it can be easily checked that the above corollary corresponds
to Theorem 4.1 in \cite{CL}, and explicit formulas for the
$\lambda$- action given at the end of section 4 in \cite{CL},
corresponds to ours in each case.
\end{remark}

%
%

\section{Lie conformal superalgebra $S_n$ and its finite irreducible
modules}\label{sec:4}

Recall that the {\it divergence} of a differential operator
$D=\sum_{i=0}^n a_i \p_i\in W(1,n)$, with $a_i\in \Lambda(1,n)$
and $\p_0=\p_t$ is defined by the formula
\begin{displaymath}
div \ D= \p_0 a_0 + \sum_{i=1}^n (-1)^{p(a_i)} \p_i a_i.
\end{displaymath}
The basic property of the divergence is $(D_1, D_2\in W(1,n))$
\begin{displaymath}
div \ [D_1, D_2] = D_1 (div\ D_2) - (-1)^{p(D_1) p(D_2)} D_2 (div
\ D_1).
\end{displaymath}
It follows that
\begin{displaymath}
S(1,n)=\{D\in W(1,n)\ : \ div\ D =0\}
\end{displaymath}
is a subalgebra of the Lie superalgebra $W(1,n)$. Similarly,
\begin{displaymath}
S(1,n)_+=\{D\in W(1,n)_+\ : \ div\ D =0\}
\end{displaymath}
is a subalgebra of $W(1,n)_+$. We have
\begin{equation}\label{eq:semi}
S(1,n) \hbox{ (resp. $S(1,n)_+$ ) } = S(1,n)' \hbox{ (resp.
$S(1,n)'_+$ ) } \oplus \CC \xi_1\cdots \xi_n \p_0,
\end{equation}
where $S(1,n)'$ (resp. $S(1,n)'_+$ )  denotes the derived
subalgebra. It is easy to see that $S(1,n)'$ is a formal
distribution Lie superalgebra, see \cite{FK}, Example 3.5.

In order to describe the associated Lie conformal superalgebra, we
need to translate the notion of divergence  to the "conformal"
language as follows. It is a $\CC[\p]$-module map $div: W_n \to
$Cur$\ \Lambda (n)$, given by
\begin{displaymath}
div\ a =\sum_{i=1} (-1)^{p(f_i)} \p_i f_i, \qquad \quad div \ f
=-\p \otimes f,
\end{displaymath}
where $a=\sum_{i=1}^n f_i \p_i\in  W(n)$ and $f\in \Lambda (n)$.
The following identity holds in $\CC[\p]\otimes \Lambda (n)$,
where $D_1, D_2\in W_n$:
\begin{equation} \label{eq:div}
div \ [{D_1}_\la D_2] = (D_1)_\la (div\ D_2) - (-1)^{p(D_1)
p(D_2)} (D_2)_{-\la-\p} (div \ D_1).
\end{equation}
Therefore,
\begin{displaymath}
S_n=\{ D\in W_n\ : \ div \ D=0\}
\end{displaymath}
is a subalgebra of the Lie conformal superalgebra $W_n$. It is
known that $S_n$ is simple for $n\geq 2$, and finite of rank $n
2^n$. Furthermore, it is the Lie conformal superalgebra associated
to the formal distribution Lie superalgebra $S(1,n)'$. The
annihilation algebra and the extended annihilation algebra is
given by
\begin{displaymath}
\A(S_n)=S(1,n)'_+ \qquad \hbox{ and } \qquad \A (S_n)^e=\CC \
ad(\p_0)\ltimes S(1,n)'_+.
\end{displaymath}

Now, we have to study representations of $S(1,n)_+$ and of its
derived algebra $S(1,n)'_+$ which has codimension 1.  Observe that
$S(1,n)_+$ inherits the $\ZZ$-gradation in $W(1,n)_+$, and
denoting by $L=S(1,n)_+$ (for the rest of this section), we have
that $L_{-1}=<\p_0, \ldots ,\p_n>$ as in $W(1,n)_+$ but the other
graded components are strictly smaller than these of $W(1,n)_+$.
Observe that $L_0= sl(1|n)$.

In order to consider weights of vectors in $S(1,n)_+$-modules, we
take the basis
\begin{displaymath}
t\p_0 +\xi_1 \p_1 , \ldots , t\p_0+ \xi_n \p_n.
\end{displaymath}
for the Cartan subalgebra. And the weights are written as
$\bar\la=(\la_1, \ldots ,\la_n)$ for the corresponding
eigenvalues.

Propositions~{\ref{prop:d1}} and ~{\ref{prop:d2}}, and
Corollary~{\ref{cor:d}} holds for $L=S(1,n)_+$ with the following
minor modification: all highest weights are the same as in the $W$
case, except for the first coordinate that should be removed.

\

Similarly, if $V$ is an $sl(1|n)$-module, then formulas
(\ref{eq:3.6}) and (\ref{eq:3.7}) define an $S_n$-module structure
in Tens$(V)$. Indeed, elements $(-1)^{p(f)} \p(f\p_i) + \p_i f$,
with $f\in \Lambda (n)$ generate $S_n$ as a $\CC[\p]$-module, and
it is  easy to see that for the action of these elements defined
by (\ref{eq:3.6}) and (\ref{eq:3.7}), one needs only $E_{ij}(v)$
for  $i\neq j$ and $(E_{00} +E_{ii})(v)$ for $i>0$.

As in the $W$-case, the classification is reduced to the study of
singular vectors in Tens$(V)$, where $V$ is an $sl(1|n)$-module.
Observe that the reduction Lemma \ref{lem:deg}   hold in this
case, and the proof is basically the same.  Therefore, analogous
computations give as the following

\begin{proposition} \label{prop:5.1}
Let $n\geq 2$ and $V$ an  irreducible finite-dimensional
$sl(1|n)$-module. If $m$ is a non-trivial singular vector in the
$S(1,n)_+$-module {\rm Tens}$\ V$ with weight $\bar\la_m$, then we
have one of the following:

\vskip .3cm

(a) $m=\xi^n \otimes v_n$, $\bar\la_m=(0,\ldots ,0, -k)$ with
$k\geq 0$,   $v_n$ is a highest weight vector in $V$ with weight
$(0,\ldots ,0, -k-1)$, and  $m$ is uniquely defined by $v_n$.

\vskip .3cm

(b) $m=\sum_{l=1}^n \xi^l \otimes v_l$, $\bar\la_m=(k,1, \ldots
,1)$ with $k\geq 2$, $v_1$ is a highest weight vector in $V$ with
weight $(k-1,1,\ldots ,1)$,  and  $m$ is uniquely defined by
$v_1$.

\vskip .3cm

(c) $m=\p (\xi_* \otimes w) + \sum_{l=1}^n \xi^l \otimes v_l$,
$\bar\la_m=(0, \ldots ,0)$,   $w$ is a highest weight vector in
$V$ with weight $(1,\ldots ,1)$, and $m$ is uniquely defined by
$w$.

\vskip .3cm

(d) $m=\p (\xi^n \otimes w) + \sum_{l=1}^{n-1} \xi_{[l,n]-\{l,n\}}
\otimes v_l + \xi^n \otimes v_n$,  $\bar\la_m=(0, \ldots ,0, -1)$,
$w$ is a highest weight vector in $V$ with weight $(1,\ldots ,1)$,
and $m$ is uniquely defined by $w$.

%
%
%
\end{proposition}

Using the above proposition, we have

\begin{theorem} \label{th:5.2} Let $L=S(1,n)_+$ ($n\geq 2$) and $F$ be an irreducible
$L_0$-module  with highest weight $\bar\la_*$. Then the
$L$-modules {\rm Ind}$_{\l0}^L F$ are irreducible finite
continuous modules except for the following cases:

\vskip .3cm

(a) $\bar\la_*=( 0,\ldots ,0,-p), p\geq 0$, where {\rm
Ind}$_{\l0}^L F= \Theta^p_+$ and the image $d^\#\Theta_+^{p+1}$ is
the only non-trivial proper submodule.

\vskip .3cm

(b) $\bar\la_*=(q ,1,\ldots ,1), q\geq 1$, where {\rm
Ind}$_{\l0}^L F= \Omega^q_-$. For $q\geq 2$ the image $d
\Omega_-^{q-1}$ is the only non-trivial proper submodule. For
$q=1$, the proper submodules are  {\rm Im}$(d)$, {\rm Ker}$(d)$
and {\rm Im}$(\alpha)$, where $\alpha$ is the composition
\begin{displaymath}
\alpha : \Theta^1_+\overset{\ d^\# \ }\longrightarrow
\Theta^0_+\simeq \Omega^0_-\overset{\ d\ }\longrightarrow
\Omega^1_-,
\end{displaymath}
and {\rm Ker}$(d)$ is the maximal proper submodule.

\vskip .3cm

\end{theorem}

\begin{proof} Similarly to the case of $W(1,n)_+$, the modules $\hbox{\rm Ind}_{\l0}^L F$
are irreducible except when they have a singular vector and the
highest weights of
 such $F$, when it could happen, are listed in (a), (b), (c) and (d) of the above
Proposition \ref{prop:5.1}. The weight $(1,\ldots , 1)$ is special
here because it is relevant to (b), (c) and (d). There are three
types of singular vectors possible in this case. The corresponding
module Ind$(F)=\Omega^1_-$ has three different submodules and all
three vectors are present. The same argument as for
$W(1,n)_+$-modules allows us easily to conclude that the listed
submodules are the only ones and the factors are irreducible.
\end{proof}

\begin{corollary}\label{cor:5.3} The theorem gives us a description of
finite continuous irreducible $S(1,n)_+$-modules when $n \geq 2$.
Such a module is either $\hbox{\rm Ind}_{\l0}^L F$ for an
irreducible finite-dimensional $L_0$-module $F$ where the highest
weight of $T$ does not belong to the types listed in (a), (b) of
the theorem or the factor of an induced module from (a), (b)   by
the submodule {\rm Ker}$(d)$.
\end{corollary}

\begin{corollary}\label{cor:5.4} The Lie superalgebras $S(1,n)_+$ and $S(1,n)'_+$ have
the same finite continuous irreducible modules, and they are
described by the previous corollary.
\end{corollary}

\begin{proof} In order to see that Theorem \ref{th:5.2} also holds for $S(1,n)'_+$, it
is basically enough  to see that  Proposition \ref{prop:5.1} holds
in this case. But, if we check the proof in the classification of
singular vectors, we see that the element $\xi_1\cdots \xi_n \p_0$
(cf. (\ref{eq:semi})) appears  only in the condition (s3) in
(\ref{eq:s3}) in the special case $g=\xi_1\cdots \xi_n$ and $j=0$.
If we track the details of the proof, we see that this special
constrain  only produces equation (\ref{eq:v8}) for $n=2$, but in
any case, this equation is not used in the rest of the proof.
 Therefore, the singular vectors are the same for both Lie superalgebras $S(1,n)'_+$ and
 $S(1,n)_+$, finishing the proof.
\end{proof}

\

Now, as in the $W_n$ case,  Theorem~{\ref{th:5.2}} and
Corollary~{\ref{cor:5.4}}, along with Section~{\ref{sec:2}} and
Propositions~{\ref{prop:1}}, ~{\ref{prop:2}} and  ~{\ref{prop:3}}
give us a complete description of finite irreducible $S_n$-modules
($n \geq 2$): it is given by Theorem \ref{th:w} in which $W_n$ is
replaced by $S_n$ and $gl(1|n)$ is replaced by $sl(1|n)$.

\begin{remark} Under the standard isomorphism between $S_2$ and small N=4 conformal
superalgebra it is easy to see that our result agrees with the
classification given in \cite{CL}. Indeed, in \cite{CL} (Theorem
6.1) the classification of irreducible modules was given in terms
of parameters $\Lambda$ and $\Delta$, and these parameters are
related to ours as follows,
\begin{eqnarray}
\lambda_1&=& -\Delta+\frac{\Lambda}{2},\\
\lambda_2&=& -\Delta-\frac{\Lambda}{2}.
\end{eqnarray}
Therefore, the case $2\Delta - \Lambda=0$ ($\Lambda\in\ZZ_+$)
corresponds to the family $\O_{\Lambda,\alpha}^*/\hbox{\rm Ker
}\tilde d^*$ of rank $4\Lambda$, and the case $2\Delta +\Lambda
+2=0$ ($\Lambda\in\ZZ_+$) corresponds to
$\O_{\Lambda+1,\alpha}/\hbox{\rm Ker }\tilde d$ of rank
$4\Lambda+8$. Therefore, we have one  module of rank 4 that
corresponds to $\Omega^*_1 / Ker\ \tilde d^*$, and by
Remark~\ref{rm:12}, the dual of this module is $\Omega_0$ (Ker is
trivial in this case) and (using Proposition~\ref{prop:d2})
$\Omega_0$ is the tensor module $Tens(V)$ where $V$ is the trivial
representation, therefore it is reducible with a maximal submodule
of codimension 1 (over $\CC$).
\end{remark}

\vskip 1cm

%
%

\section{Lie conformal superalgebras $S_{n,b}$ and $\tilde S_n$, and their
finite irreducible modules}\label{sec:5}

\noindent{\it  - Case $S_{n,b}$:}

\

\noindent For any $b\in \mathbb C$, $b\neq 0$, we take
\begin{displaymath}
S(1,n,b)=\{D\in W(1,n) | div (e^{bx} D)=0\}.
\end{displaymath}
This is a formal distribution subalgebra of $W(1,n)$. The
associated Lie conformal superalgebra is constructed explicitly as
follows. Let $D=\sum_{i=1}^n P_i(\p , \xi) \p_i + f(\p, \xi)$ be
an element of $W_n$. We define the deformed divergence as
\begin{displaymath}
div_bD=div D + bf.
\end{displaymath}
It still satisfies equation \ref{eq:div}, therefore
\begin{displaymath}
S_{n,b}= \{D\in W_n | div_b D=0\}
\end{displaymath}
is a subalgebra of $W_n$, which is simple for $n\geq 2$ and has
rank $n2^n$. Since $S_{n,0} = S_n$ has been discussed in the
previous section, we can (and will) assume that $b \neq 0$.

If $b\neq 0$, the extended annihilation algebra is given by
\begin{displaymath}
(Alg (S_{n,b}))^+=\mathbb C ad(\p_0 - b \sum_{i=1}^n \xi_i \p_i)
\ltimes S(1,n)_+\simeq CS(1,n)'_+
\end{displaymath}
where $CS(1,n)'_+$ is obtained from $S(1,n)'_+$ by enlarging
$sl(1,n)$ to $gl(1,n)$ in the 0th-component.

Therefore, the construction of all finite irreducible modules over
$S_{n,b}$ is the same as that for $W_n$, but without   twisting by
$\alpha$. Hence, using Theorem \ref{th:w}, we have

\begin{theorem} \label{th:sb}
The following is a complete list of finite irreducible
$S_{n,b}$-modules $(n\geq 2, b\in \CC , b \neq 0 )$:
\alphaparenlist
\begin{enumerate}
\item 
{\rm Tens}$ V$, where $V$ is a finite-dimensional irreducible
$gl(1|n)$-module different from $\Lambda^k(\CC^{1|n})^*, k=1,2,
\ldots $ and $\Lambda^k(\CC^{1|n})$,  $k=0,1,2,...$,
\item 
$\O_{k}^*/\hbox{\rm Ker }\tilde d^*, k=1,2,\ldots$, and the same
modules with reversed parity,
\item 
$S_{n,b}$-modules dual to $(b)$, with $k>1$.
\end{enumerate}
\end{theorem}

\vskip 1cm

\noindent{\it - Case $\tilde S_{n}$:}

\

\noindent Let $n\in\ZZ_+$ be an even integer. We take
\begin{displaymath}
\tilde S(1,n)=\{D\in W(1,n) | div((1+\xi_1\dots \xi_n)D)=0 \}.
\end{displaymath}
This is a formal distribution subalgebra of $W(1,n)$. The
associated Lie conformal superalgebra $\tilde S_n$ is constructed
explicitly as follows:
\begin{displaymath}
\tilde S_n=\{D\in W_n | div((1+\xi_1\dots
\xi_n)D)=0\}=(1-\xi_1\dots \xi_n)S_n.
\end{displaymath}
The Lie conformal superalgebra $\tilde S_n$ is simple for $n\geq
2$ and has rank $n2^n$.

The extended annihilation algebra is given by
\begin{displaymath}
(Alg (\tilde S_{n}))^+=\mathbb C ad(\p_0 - \xi_1\dots\xi_n \p_0)
\ltimes S(1,n)'_+ \simeq  S(1,n)_+.
\end{displaymath}

Therefore, the construction of all finite irreducible modules over
$\tilde S_{n}$ is the same as that for $S_n$, but without the
twist by $\alpha$.

\vskip 1cm

\noindent{\bf Acknowledgment.} C. Boyallian and J. Liberati were
supported in part by grants of Conicet, ANPCyT, Fundaci\'on
Antorchas, Agencia Cba Ciencia, Secyt-UNC and Fomec (Argentina).
Special thanks go to MSRI (Berkeley) for the hospitality during
our stay there, where part of this work started.



\noindent{\bf Authors'addresses}

\

\noindent Department of Mathematics, M.I.T., Cambridge, MA 02139,
USA,

\noindent email: kac@math.mit.edu.

\

\noindent Famaf-CIEM, Ciudad Universitaria, (5000) Cordoba,
Argentina,

\noindent email: boyallia@mate.uncor.edu, liberati@mate.uncor.edu.

\

\noindent Department of Mathematics, NTNU, Gl\o shaugen, N-7491
Trondheim, Norway,

\noindent email: rudakov@math.ntnu.no

\end{document}